\documentclass{llncs}
\usepackage{microtype}
\usepackage{gastex,graphicx}
\usepackage{amsmath}
\usepackage{amsfonts}
\usepackage{amssymb}
\usepackage{tikz}
\newcommand{\until}{\:\mathcal{U}}
\newcommand{\since}{\:\mathcal{S}}
\newcommand{\R}{\:\mathbb{R}}
\newcommand{\N}{\:\mathbb{N}}
\newcommand{\fut}{\Diamond}
\newcommand{\past}{\mbox{$\Diamond\hspace{-0.27cm}-$}}

\mathchardef\mhyphen="2D
\mathchardef\mhyph="2D

\newcommand{\nex}{\mathcal{O}}

\newcommand{\mtlfull}{\mathsf{MTL}^{pw}[\until_I, \since_{I}]}
\newcommand{\mtlfullunary}{\mathsf{MTL}^{pw}[\fut_I, \past_{I}]}
\newcommand{\mtlu}{\mathsf{MTL}^{pw}[\until_I,\since]}
\newcommand{\mtl}{\mathsf{MTL}^{pw}[\until_I]}
\newcommand{\mtlfut}{\mathsf{MTL}^{c}[\fut_I]}
\newcommand{\mtluc}{\mathsf{MTL}^{c}[\until_I]}
\newcommand{\mtlfutpw}{\mathsf{MTL}^{pw}[\fut_I]}
\newcommand{\mtlfutp}{\mathsf{MTL}^{pw}[\fut, \past_{I}]}
\newcommand{\mtlsns}{\mathsf{MTL}^{pw}[\until_I,\since_{NS}]}
\newcommand{\mtluns}{\mathsf{MTL}^{pw}[\until_{NS},\since_I]}
\newcommand{\mitl}{\mathsf{MTL}^{pw}[\until_{NS},\since_{NS}]}
\newcommand{\mitlfp}{\mathsf{MTL}^{pw}[\fut_{NS},\past_{NS}]}

\newcommand{\gmtlfull}{\mathsf{MTL}[\until_I, \since_{I}]}
\newcommand{\gmtlfullunary}{\mathsf{MTL}[\fut_I, \past_{I}]}

\newcommand{\gmtl}{\mathsf{MTL}[\until_I]}

\newcommand{\gmitl}{\mathsf{MTL}[\until_{NS},\since_{NS}]}

\newcommand{\oomit}[1]{}

\begin{document}

\title{On the Decidability and Complexity of Some Fragments of Metric Temporal Logic}

\author{Khushraj Madnani$^1$, Shankara Narayanan Krishna$^{1}$ and Paritosh K. Pandya$^{2}$}
\institute{
$^1$IIT Bombay, Powai, Mumbai, India 400 076\\
$^2$Tata Institute of Fundamental Research, Colaba, Mumbai, India 400 \\
\email{\{khushraj,krishnas\}@cse.iitb.ac.in, pandya@tcs.tifr.res.in}
}

\maketitle

\begin{abstract}
Metric Temporal Logic, $\mtlfull$, is amongst the most studied
real-time logics. It exhibits considerable diversity in expressiveness and
decidability properties based on the permitted set of modalities and the
nature of time interval constraints $I$. 
In this paper, we sharpen the decidability results by showing that the
satisfiability of $\mtlsns$ (where $NS$ denotes non-singular intervals) 
is also decidable over finite pointwise
strictly monotonic time. We give a satisfiability preserving reduction from
$\mtlsns$ to the decidable logic $\mtl$ of Ouaknine and Worrell
using a novel technique of temporal projections with oversampling. 
We also investigate the decidability of unary
fragment $\mtlfullunary$ 
and we compare the expressive powers of some of these fragments. 
\end{abstract}

\section{Introduction}
Real time logics specify properties of how system state evolves with time and where
quantitative  time distance between events is significant.
Metric Temporal Logic ($\mathsf{MTL}$) introduced by Koymans \cite{Koymans90} is a prominent real-time logic. In this logic, the temporal modalities $\until_I$ and $\since_I$ are constrained 
by a time interval $I$.  $\mathsf{MTL}$  exhibits considerable diversity in expressiveness and
decidability 
 based on the permitted set of modalities and the
nature of time interval constraints $I$. 

The classical results of Alur and
Henzinger showed that satisfiability of $\gmtlfull$ as well as its model checking problem against timed automata are undecidable in general \cite{AlurH91}, \cite{Hen91}. In their seminal paper \cite{AFH96}, the authors proposed
a sublogic $\gmitl$ having only non-singular intervals $NS$, where the satisfiability is decidable with $\mathit{EXPSPACE\--complete}$ complexity. 
The satisfiability of $\gmtl$ was considered to be undecidable for  a long time, until Ouaknine and Worrell \cite{OuaknineW05} proved that 
the satisfiability of $\mtl$ over finite timed words is decidable, albeit with a non-primitive recursive lower bound. 
Subsequently, in \cite{OuaknineW06}, it was shown that over infinite timed words, the satisfiability  of $\mtl$ is undecidable. The satisfiability of $\mtluc$ over continuous time is also undecidable.

In this paper, we sharpen the decidability results for fragments of $\mtlfull$.
We consider the logic $\mtlsns$ which has been shown \cite{concur11} to be strictly more  expressive 
than the known decidable fragments $\gmitl$ as well as $\gmtl$, but strictly less expressive than the full $\gmtlfull$. 
As our main result, we show that satisfiability of $\mtlsns$ is decidable  over
pointwise strictly monotonic time (i.e. finite timed words).
By symmetry it is easy to show that $\mtluns$ also has decidable satisfiability.
The result is established by giving a satisfiability preserving reduction from
the logic $\mtlsns$ to $\mtl$ using a novel technique of
{\em temporal projections with oversampling}. 

Temporal projection is a technique which allows  obtaining equi-satisfiable formula
in a more restricted logic by using additional auxiliary propositions. The formula transformation is carried out in a conservative fashion
so that the models of the original formula and those of the transformed formula are related in projection-embedding  fashion. This technique was used in a number of works on continuous time temporal logics to obtain  equi-satisfiable formulae with only restricted set of modalities (e.g. from $\gmtlfull$ to $\gmtl$).
\cite{fundinfo04,deepak08,PrabhakarD06,formats11}. 
In this paper, we generalize the technique to pointwise models where transformed formula is  interpreted over timed words which are ``oversampling'' of the original timed words. Thus, model embedding involves  adding intermediate ``non-action'' time points in the timed word where only the auxiliary propositions are interpreted.

Our transformation of the $\mtlsns$ formula to $\mtl$ formula relies upon the properties of the  unary modality $\past_{NS}$. In a recent work \cite{atva12},  Pandya and Shah formulated  a ``horizontal stacking'' properties of the
bounded $\past_{[l,u]} \phi$ modality: this allows the truth of  $\past_{[l,u]} \phi$ at a point in any unit interval 
to be related to the first and the last occurrences  of $\phi$ in some related unit intervals. 
Our transformation builds upon these properties to achieve elimination of past modalities using temporal projection with oversampling.

Several real-time properties can be specified using only the unary 
future $\fut_I$, and past $\past_I$ operators. As our second main contribution,
we investigate the decidability of unary fragment $\gmtlfullunary$ 
(this question was posed by A. Rabinovich in a personal communication).
We show that $\mtlfullunary$ over finite pointwise time is undecidable, whereas 
$\mtlfutpw$ over finite pointwise models already has Ackermann-hard  satisfiability checking.
Hence, restriction to unary modalities does not improve the decidability properties of $\mtlfull$.
We  compare the expressive powers of some of these fragments using the
technique of EF games for $\mathsf{MTL}$ \cite{concur11}.  
\section{Metric Temporal Logic}
In this section, we describe the syntax and semantics of $\mathsf{MTL}$ in the pointwise sense. The definitions 
below are standard.  Let $\Sigma$ be a finite alphabet of events. A finite timed word over $\Sigma$ is a sequence 
$\rho=(A_1, t_1)(A_2, t_2) \dots (A_n, t_n)$ where $A_i \subseteq \Sigma
$, $A_i \cap \Sigma \neq \emptyset$, and $t_i \in \R_{\geq 0}$
 for $1 \leq i \leq n$. Further, $t_i < t_{j}$ for all $1 \leq i < j \leq n$.
$BP$ and $EP$ are special endmarkers which hold at the start and end point 
of any timed word. We use the short form $\rho=(\sigma, \tau)$ to represent a timed word; 
$\sigma=A_1A_2 \dots A_n$ and $\tau=t_1t_2 \dots t_n$. Each point $\rho(i)$, $1 \leq i \leq n$ 
is called an action point, and let $dom(\rho)=\{1,2,\dots,n\}$ be the set of positions of $\rho$.
Let $time(\rho)$ denote the time stamp of the last action point of $\rho$. 
Let $\sigma_i=A_i$ and $\tau_i=t_i$. 
Given $\Sigma$,  the formulae of $\mathsf{MTL}$ are built from $\Sigma$, modalities $BP,EP$ using boolean connectives and 
time constrained versions of the modalities $\until$ and  $\since$ as follows: \\
$~~~~~~~~~~~~~~\varphi::=a (\in \Sigma)~|~BP~|~EP~|true~|\varphi \wedge \varphi~|~\neg \varphi~|
~\varphi \until_I \varphi~|~\varphi \since_I \varphi$\\
where  $I$ is an open, half-open or closed interval 
with end points in $\N \cup \{\infty\}$.    

\noindent{Point-Wise Semantics} : 
\label{point}
Given a finite timed word $\rho$, and an $\mathsf{MTL}$ formula $\varphi$, in the pointwise semantics, 
the temporal connectives of $\varphi$ quantify over a countable set of positions 
in $\rho$.  For an alphabet $\Sigma$, a timed word $\rho=(\sigma, \tau)$, a position 
$i \in dom(\rho)$, and an $\mathsf{MTL}$ formula $\varphi$, the satisfaction of $\varphi$ at a position $i$ 
of $\rho$ is denoted $(\rho, i) \models \varphi$, and is defined as follows:

\begin{tabular}{l c l c l c l}
$\rho, i \models a$  & \; $\leftrightarrow$ & \; $a \in \sigma_{i}$\\
$\rho,i  \models \neg \varphi$ & \; $\leftrightarrow$  & \; $\rho,i \nvDash  \varphi$\\
$\rho,i \models \varphi_{1} \wedge \varphi_{2}$  & \; $\leftrightarrow$  & \; $\rho,i \models \varphi_{1}$ 
and $\rho,i\ \models\ \varphi_{2}$\\
$\rho,i\ \models\ \varphi_{1} \until_{I} \varphi_{2}$  & \; $\leftrightarrow$  & \; $\exists j > i$, 
$\rho,j\ \models\ \varphi_{2}$, $t_{j} - t_{i} \in I$,  and  $\rho,k\ \models\ \varphi_{1}$ $\forall$ $i< k <j$\\
$\rho,i\ \models\ \varphi_{1} \since_{I} \varphi_{2}$  & \; $\leftrightarrow$  & \; $\exists\ j < i$, 
 $\rho,j\ \models\ \varphi_{2}$, $t_{i} - t_{j} \in I$,  and  $\rho,k\ \models\ \varphi_{1}$ $\forall$  $j<k < i$\\
\end{tabular}

$\rho$ satisfies $\varphi$ denoted $\rho \models \varphi$ 
iff $\rho,0 \models \varphi$. Let $L(\varphi)=\{\rho \mid \rho, 0 \models \varphi\}$. 
The set of all timed words over $\Sigma$ is denoted $T\Sigma^*$. 
Additional temporal connectives are defined in the standard way: 
we have the constrained future and past eventuality operators $\fut_I a \equiv true \until_I a$ and 
$\past_I a \equiv true \since_I a$, and their duals  
$\Box_I a \equiv \neg \fut_I \neg a$, 
$\boxminus_I a \equiv \neg \past_I \neg a$. Weak versions of  operators 
are defined as : $\widetilde{\fut} a=a \vee \fut a, 
\widetilde{\Box} a= a \wedge \Box a$, $a \widetilde{\until} b= b \vee (a \until b)$.
We denote by $\mathsf{MTL}^{pw}[W]$ the 
class of $\mathsf{MTL}$ formulae interpreted in the pointwise semantics
having modalities $W \subseteq \{\fut_I$ or $\until_I, \past_I$ or $\since_I\}$.
If $I$ is a non-singular interval of the form $\langle a, b \rangle$ with $a \neq b$ 
then we denote the modalities by $\fut_{NS}$ ($\until_{NS}$), and $\past_{NS}$ ($\since_{NS}$). 

\section{ Decidability of $\mtlsns$}
In this section, we show that the class $\mtlsns$ is decidable, by giving a satisfiability preserving 
reduction  to $\mtl$.  Two important transformations needed in this reduction
are flattening of the formula (which removes nesting of temporal operators using
auxiliary propositions) and oversampling closure which makes the
truth of the formula invariant under insertion of additional time points (oversampling).
\oomit{Given  $\varphi \in \mtlsns$ over $\Sigma$, 
we synthesize a formula $\varphi_U \in \mtl$ over $\Sigma' \supseteq \Sigma$ 
such that $\varphi$ is satisfiable iff $\varphi_U$ is satisfiable.}

 \subsection{Temporal Projections}
\label{remove-past}
Let $\Sigma' \supseteq \Sigma$.  
Consider a word $\rho'=(Y_1,0)(Y_2,t'_2) \dots (Y_m,t'_m)$ in $T\Sigma'^*$.  
Then, $\rho' \upharpoonright \Sigma$ is the timed word 
 $\rho \in T\Sigma^*$ obtained by the steps (E1) followed by (E2): \\ 
 (E1) Erase all symbols of $\Sigma' \backslash \Sigma$ from $\rho'$. 
  Call the resultant word $\rho''$.\\
  (E2) Erase all symbols of the form $(\emptyset, t_i)$ from $\rho''$, to obtain $\rho$.
Given $\rho'$ as above, a point $\rho'(i)$ is called 
an {\bf action-point} iff $Y_i \cap \Sigma \neq \emptyset$. 
For $\rho=(X_1,0)(X_2,t_2)\dots(X_n,t_n)$ a timed word in $T\Sigma^*$,
$\rho'$ is called a $\Sigma'$-{\bf oversampling} of $\rho$, denoted by $\rho' \Downarrow^{\Sigma}_{\Sigma'}=\rho$
iff (i) $\rho' \upharpoonright \Sigma = \rho$, and (ii) $1,m$ are action points.  
In this case $\rho$ is called a projection of 
$\rho'$. A $\Sigma'$-oversampling $\rho'$ is called a {\it \bf simple extension} of $\rho$, 
denoted by $\rho' \downarrow^{\Sigma}_{\Sigma'}=\rho$ iff $m=n$. 
For a $\Sigma'$-oversampling $\rho'$ of $\rho$, we  define a map $g$ from the action points of $\rho'$ 
 to the points of $\rho$ as follows: 
 $g(\rho'(1))=\rho(1)$.  For $1 < i < m$, $g(\rho'(i))=\rho(j)$ iff 
 (i) $Y_i \cap \Sigma= X_j$, $t'_i=t_j$, and (ii) For some $h >0$, $g(\rho'(i-h))=\rho(j-1)$ iff 
 for all $i-h < r < i$, $\rho'(r)$ are non-action points, and $Y_{i-h}\cap \Sigma=X_{j-1}$. 
For simple extension, this mapping is the identity. 

The above notions can be extended to timed languages. For $L \subseteq T\Sigma^*$ and $L' \subseteq T\Sigma'^*$,
$L'\Downarrow^{\Sigma}_{\Sigma'}=L$ iff for every $\rho' \in L'$, we have some $\rho \in L$ such that 
$\rho'\Downarrow^{\Sigma}_{\Sigma'}=\rho$, and for every $\rho \in L$, there exists some $\rho' \in L'$
such that $\rho'\Downarrow^{\Sigma}_{\Sigma'}=\rho$.

\begin{example}
 Let $\Sigma=\{a,b\}$ and let $\Sigma'=\{a,b,c\}$. 
 A $\Sigma'$-oversampling of \\
 $\rho=(\{a\},0)(\{a,b\},0.3)(\{a\},4.8)$
is $\rho_1=(\{a,c\},0)(\{c\},0.1)(\{a,b\},0.3)(\{a,c\},4.8)$. 
$\rho_2=(\{a,c\},0)(\{b,c\},0.1)(\{a,b\},0.3)(\{a,c\},4.8)$ is not a $\Sigma'$-oversampling of $\rho$. 
$\rho_1 \upharpoonright \{a,b\}=\rho$, 
while $\rho_2 \upharpoonright \{a,b\} \neq \rho$.
Neither $\rho_1$ nor $\rho_2$ is simple extension of $\rho$. Also, $g(\rho_1(1))=\rho(1), 
g(\rho_1(3))=\rho(2), g(\rho_1(4))=\rho(3)$.
 \end{example}

 \noindent{\bf Flattening} Let $\varphi \in \mtlfull$ over $\Sigma$.
Given any sub-formula $\psi_i$ of $\varphi$, 
and a fresh symbol $a_i$, 
the equivalence $X_i=\widetilde{\Box}(\psi_i  \Leftrightarrow a_i)$
is called a temporal definition and  $a_i$ is called a witness. 
Let $\psi=\varphi[a_i /\psi_i]$ be the formula 
obtained by replacing all occurrences of $\psi_i$ in $\varphi$, with the witness $a_i$. 
Flattening is done recursively until we have replaced all future/past modalities of interest
with witness variables, obtaining $\varphi_{flat}=\psi \wedge X$, where $X$ is the conjunction 
of all temporal definitions. Let $\Sigma'$ be the 
set of auxiliary witness propositions used in $X$, and let $\Delta=\Sigma \cup \Sigma'$. 
Let $act=\bigvee_{a \in \Sigma}a$. 
\begin{lemma}
 \label{lem:flat}
Let $\varphi \in \mtlfull$ over $\Sigma$ and
let $\varphi_{flat}$ be the flattened formula obtained from $\varphi$.
Note that  $\varphi_{flat}$ is over $\Delta$.
Let $\rho$ be a timed word in $T\Sigma^*$. 
 For any $\Delta$-extension $\rho'$ of $\rho$ such that 
 $\rho'\downarrow^{\Sigma}_{\Delta}=\rho$, 
 we have $\rho \models \varphi$  iff $\rho' \models \varphi_{flat}$.  
 \end{lemma}
 
 \begin{lemma}[oversampling closure]
\label{lem:gen}
Let $\varphi \in \mtlfull$ over $\Sigma$ and $\Sigma' \supseteq \Sigma$.
Then $L(\varphi)=L(\hat{\varphi})\Downarrow^{\Sigma}_{\Sigma'}$ where 
$\hat{\varphi}=\varphi' \wedge ((BP \Rightarrow act) \wedge (EP \Rightarrow act))$, and $\varphi'$ is obtained 
from $\varphi$ by replacing recursively all subformulae of the form $(a_i \until_I a_j)$ [$(a_i \since_I a_j)$] with 
  $(act \Rightarrow a_i) \until_I (a_j\wedge act)$ [$(act \Rightarrow a_i)\since_I (a_j \wedge act)$]. 
\end{lemma}
\begin{proof}
In Appendix \ref{proof:gen}.
\end{proof}

 \begin{lemma}
\label{lem:flatgen}
Let $\varphi \in \mtlfull$ over $\Sigma$, and $\Sigma' \supseteq \Sigma$.
Then $L(\varphi)=L(\hat{\varphi}_{flat})\Downarrow^{\Sigma}_{\Sigma'}$, where $\hat{\varphi}_{flat}
$ is obtained from $\varphi_{flat}$ using Lemma \ref{lem:gen}. 
\end{lemma}

Given a formula $\varphi$ (over $\Sigma$) of logic $L_1$,
we can often find a formula $\psi$ (over $\Sigma'$)
of a much simpler/desirable  logic $L_2$  such that $L(\varphi) = L(\psi)\Downarrow^{\Sigma}_{\Sigma'}$.
We say that $\varphi$ is equivalent modulo temporal projections (equisatisfiable) to $\psi$. 
This is denoted $\psi \Downarrow \Sigma \equiv \varphi$. 
Example  \ref{ex-past} illustrates this.

\begin{example}
\label{ex-past}
Consider the formula $\varphi=\fut(a \wedge \past_{(1,\infty)} c)$ over
$\Sigma=\{a,c\}$. $act=a \vee c$.  Flattening $\varphi$ gives the formula 
 $\varphi_{flat}= \fut \alpha \wedge X$, with $X=\widetilde{\Box}[\alpha \Leftrightarrow (a \wedge \beta)] 
 \wedge \widetilde{\Box}[\beta \Leftrightarrow \past_{(1,\infty)} c]$.
 $\hat{\varphi}_{flat}=\fut(\alpha \wedge act) \wedge \hat{X}$, with 
  $\hat{X}=\widetilde{\Box}[act \Rightarrow (\alpha \Leftrightarrow (a \wedge \beta))] 
  \wedge \widetilde{\Box}[act \Rightarrow (\beta \Leftrightarrow \past_{(1,\infty)}(c \wedge act))]$. 
We now replace the past subformula $\widetilde{\Box}[act \Rightarrow (\beta \Leftrightarrow \past_{(1,\infty)}(c \wedge act))]$
 of $\hat{X}$ with the formula $\nu \in \mtl$: 
 $[(act \Rightarrow (\neg c \wedge \neg \beta)) \widetilde{\until}  
 [(c \wedge act) \wedge \widetilde{\Box}_{[0,1]}(act \Rightarrow \neg \beta)]] 
\wedge \widetilde{\Box}[(c \wedge act) \Rightarrow \Box_{(1,\infty)}(act \Rightarrow \beta)]$.
Then, the formula $\fut(\alpha \wedge act) \wedge \widetilde{\Box}[act \Rightarrow (\alpha \Leftrightarrow (a \wedge \beta))]   \wedge \nu$ is equivalent to $\hat{\varphi}_{flat}$. 
 \end{example}

\subsection{Elimination of Past}
\label{past-elim}
We now give a satisfiability preserving reduction from 
$\mtlsns$ to logic $\mtl$.  
Fix a formula $\phi \in \mtlsns$ over propositions $\Sigma$. 
For simplicity we assume that $S_{NS}$ only has intervals $NS$ which are left-closed-right-open, 
e.g. $[3,\infty)$ or $[5,17)$. Other forms of intervals can be 
handled similarly using the reduction given below. 
\begin{itemize}
\item Given $\phi \in \mtlsns$ we transform the formula to an equivalent
formula $\phi' \in {\mathsf MTL}^{pw}[\until_I,\past_J]$ where $J$ is either infinite (i.e. $[l,\infty)$)
or unit (i.e. $[l,l+1)$). Standard techniques (see  \cite{PrabhakarD06,formats11}) apply to give this reduction.
\item Removal of $\past_{[l,l+1)}$ modality requires us to consider behaviours
where additional non-action time points have to be added. Each occurrence of the
$\past$ operator gives its own requirement of adding time points. Hence, we consider
equisatisfiable $\hat{\varphi}_{flat}$ which is invariant under such oversampling
(by Lemma \ref{lem:flatgen}).
\item Let $\hat{\varphi}_{flat}$ over $\Sigma'$ be obtained by flattening and
oversampling closure as in Lemma \ref{lem:flatgen}. 
This formula consists of a conjunction of temporal definitions. 
Lemma \ref{remove-pastinf} below shows how  temporal definition with past operator of the form
$\widetilde{\Box}[act \Rightarrow (b \Leftrightarrow \past_{[l, \infty)}(a \wedge act))]$
can be replaced by an {\it equivalent} formula in $\mtl$.
 Similarly, Lemma \ref{past-b1} gives elimination of temporal definition involving $\past_{[l,l+1)}$ operator using an {\em equi-satisfiable} $\mtl$ formula.
\item The above constitute
the main lemmas of our proof. By repeatedly applying them we get an equi-satisfiable formula of $\mtl$.
\end{itemize}

\begin{lemma}
\label{remove-pastinf}
Consider a temporal definition 
 $\hat{X}_{[l, \infty)}=\widetilde{\Box}[act \Rightarrow (b \Leftrightarrow \past_{[l, \infty)}(a \wedge act))]$. 
 Then we can synthesize a formula $\nu \in \mtl$ equivalent to $\hat{X}_{[l, \infty)}$.
\end{lemma}
\begin{proof}
 Let $\alpha=(act \Rightarrow (\neg a \wedge \neg b))$.
Consider the following formulae in $\mtl$:
 \begin{enumerate}
  \item $\varphi_1: [\widetilde{\Box} \alpha \vee \{\alpha \widetilde{\until}[(a \wedge act) \wedge \widetilde{\Box}_{[0,l)}(act \Rightarrow \neg b)] \}]$
\item $\varphi_2: \widetilde{\Box}[(a \wedge act) \Rightarrow \Box_{[l,\infty)}(act \Rightarrow b)]$.
 \end{enumerate}
 Let $\nu=\varphi_1 \wedge \varphi_2$.  
 We claim that $\rho' \models \hat{X}_{(l,\infty)}$ iff $\rho' \models \nu$ for any $\rho' \in T\Sigma''^*$.
  
Assume $\rho' \models \hat{X}_{(l, \infty)}$. 
Assume to contrary that $\rho' \models \neg \varphi_1$. Then, either there is a point marked $act \wedge b$ 
 before the first occurrence of $a \wedge act$, or there 
is a point marked $act \wedge b$ in the $[0,l)$ future of the first $a \wedge act$. 
Both of these imply $\neg \hat{X}_{[l, \infty)}$ giving contradiction. Assume to contraty that  $\rho' \models \neg \varphi_2$, then 
some point $act$ in the $[l, \infty)$ future of a certain $a\wedge act$ 
is marked $\neg b$, which again contradicts $\hat{X}_{[l, \infty)}$.
 Hence $\rho' \models \nu$.  The converse can be found in Appendix \ref{pastinf-c}. 
 \end{proof}

Next, consider a past formula of the form $\past_{[l, l+1)}$.
The following lemma \cite{atva12} gives conditions under 
which a formula $\past_{\langle l, l+1 \rangle}a$ holds at 
a point $i$ with $\tau_i\in[t+l+1,t+l+2), t,l \in \N$ of a timed word. 
The truth of $\past_{\langle l, l+1 \rangle}a$ at $\tau_i$ 
depends on the first and last points marked $a$ in the intervals $[t, t+1)$ and 
$[t+1, t+2)$.   These are denoted by ${\mathcal F}^a_{[t,t+1)}$ and 
${\mathcal L}^a_{[t,t+1)}$ in the lemma.  
Figures \ref{case1}--\ref{case4} depict the
regions where $\past_{[l, l+1)}$ holds (these are denoted by $b$).

\begin{lemma}[\cite{atva12}]
\label{atva12}
 Given a timed word $\rho=(\sigma, \tau)$, integers $l, t$ and an  point $i \in dom(\rho)$.
For $\tau_i \in [t+l+1,t+l+2)$, we have
 $\rho, i \models \past_{\langle l, l+1 \rangle}a$ iff 
 \begin{itemize}
  \item $\tau_i > {\mathcal F}^a_{[t,t+1)}(\rho) \wedge \tau_i \in [t+l+1, {\mathcal L}^a_{[t,t+1)}(\rho)+l+1\rangle$, or
  \item $\tau_i > {\mathcal F}^a_{[t+1,t+2)}(\rho) \wedge \tau_i \in \langle {\mathcal F}^a_{[t+1,t+2)}(\rho)+l, t+l+2)$.
 \end{itemize}
\end{lemma}


Consider a temporal definition, 
$\hat{X}_{[l,l+1)}=\widetilde{\Box}[act \Rightarrow (b \Leftrightarrow \past_{[l, l+1)}(a \wedge act))]$ whose 
defining modality is $\past_{[l, l+1)}$.
In Lemma \ref{past-b1}, we show how to synthesize a formula $\psi \in \mtl$ which 
is equisatisfiable to $\hat{X}_{[l,l+1)}$. 
For this, we construct  an oversampling $\rho'$ of $\rho$ over an extended alphabet $\Sigma''$.

\begin{lemma}
 \label{past-b1}
 Consider a temporal definition $\hat{X}_{[l, l+1)}=\widetilde{\Box}[act \Rightarrow (b \Leftrightarrow \past_{[l,l+1)}(a \wedge act))]$. Let $a,b \in \Sigma'$. 
 Let $\Delta=\Sigma' \cup\{c,beg_b,end_b\}$. 
 Then we can synthesize $\psi \in \mtl$ over $\Delta$ such that 
 for any  $\Sigma'' \supseteq \Delta$ and $\rho' \in T\Sigma''^*$, 
 \begin{enumerate}
  \item $\rho' \models \psi \Rightarrow \rho' \models \hat{X}_{[l,l+1)}$
  \item $\rho' \models \hat{X}_{[l,l+1)} \Rightarrow \exists \rho'' \in T\Sigma''^*$ such that 
$\rho'' \models \psi$, and $\rho'' \Downarrow^{\Sigma'}_{\Sigma''}= \rho' \Downarrow^{\Sigma'}_{\Sigma''}$.
\end{enumerate}
\end{lemma}
\begin{proof}
Firstly, notice that if there exists $i$ in $dom(\rho)$ marked $act \land a$,  then 
 all points $j$ in $dom(\rho)$ marked $act$ such that $t_j \in [t_i+l, t_i+l+1)$ must be marked $b$. 
 This is enforced by the following formula:
\begin{itemize}
 \item $\varphi_{8}: \widetilde{\Box}[(a\wedge act) \Rightarrow \Box_{[l,l+1)}(act \Rightarrow b)]$
\end{itemize}
$\varphi_8$ enforces the direction $act \Rightarrow (\past_{[l,l+1)} (a \wedge act) \Rightarrow b)$ of $\hat{X}_{[l,l+1)}$.
  Marking points with $\neg b$ is considerably more involved.  At a time point $t_i \in [t+l+1, t+l+2)$, $t \in \N$, 
$b$ holds only if there is an $a \wedge act$ in 
  the interval $(t_i-l-1, t_i-l] \subseteq [t, t+2)$.  
  Here we exploit Lemma \ref{atva12}. But to state its conditions using only future modalities, we need auxiliary propositions $c,beg_b, end_b$  which are required to hold at some possibly non-action points. 
  Proposition $c$  marks every integer valued time point within the time span of $\rho$. 
  The following formula specifies the behaviour of $c$.  Note that $c$ is uniquely  determined by the formula.
\begin{itemize}
 \item $\varphi_1: c \wedge \widetilde{\Box}[ (c \wedge \neg EP) \Rightarrow
\Box_{(0,1)}\neg c \wedge[\fut_{[1,1]} c \vee \fut_{(0,1)}EP]]$
\end{itemize}
To see the need for $beg_b, end_b$, consider the case for some $t \in \N$, 
 where the last $act \wedge a$ in $(t,t+1]$ occurs at $u$ and the first $act \wedge a$ 
 in $[t+1, t+2)$ occurs at $v$. If $v-u >1$, then all points $act$ in $[u+l+1, v+l)$ must be marked $\neg b$. 
 See Figure \ref{case4}.
To facilitate this marking correctly, we introduce  
 a non-action point marked $end_b$  at $v+l$, and a non-action point 
 marked $beg_b$ at $u+l+1$ in $\rho'$, and 
 state that $\neg b$ holds at all action points between 
 $beg_b$ and $end_b$.  The following formulae 
 assert that $end_b$ holds at distance $l$  from the first $a$ 
 in each unit interval with integral end points. 
 The first such $end_b$ happens beyond $[0,l)$:
\begin{itemize}
 \item $\varphi_2$ : $\widetilde{\Box}[(c \wedge \widetilde{\fut}_{[0,1)}(a \wedge act))
 \Rightarrow 
 [(act \Rightarrow \neg a) \widetilde{\until}_{[0,1)} 
 ((a\wedge act) \wedge [\widetilde{\fut}_{[l,l]}end_b \vee \widetilde{\fut}_{[0,l)}EP])]]$ 
 \item $\varphi_3$: $\widetilde{\Box}_{[0,l)}\neg end_b$ (if $l\ne 0$)
\end{itemize}
  The following formulae assert that 
 that $beg_b$ holds at distance $l+1$  from the last $a$ 
 in each unit interval with integral end points. 
 The first such $beg_b$ happens beyond $[0,l+1)$:
 \begin{itemize}
  \item $\varphi_4$ : $\widetilde{\Box}([c \wedge \widetilde{\fut}_{[0,1)}(a\wedge act)] \Rightarrow$
 $~ \widetilde{\fut}_{[0,1)} \{(a\wedge act) \wedge [((act \Rightarrow \neg a
 ) \wedge \neg c)\until c] \wedge $\\
 $(\fut_{[l+1,l+1]}beg_b \vee \widetilde{\fut}_{[0,l+1)}EP)\}) $
\item $\varphi_5 : \widetilde{\Box}_{[0,l+1)}\neg beg_b$ 
 \end{itemize}
The following formulae assert that each unit interval with integral end points 
  can have atmost one $end_b$, and one $beg_b$ : if a unit interval $[t, t+1)$, with $t \in \N$ 
  has no $a$, then there is no $end_b$ in the interval $[t+l, t+l+1)$, and there is no $beg_b$
   in the interval $[t+l+1, t+l+2)$.
   \begin{itemize}
    \item $\varphi_6$ : $c \wedge \widetilde{\Box}_{[0,1)}(act \Rightarrow \neg a) 
 \Rightarrow \widetilde{\Box}_{[l, l+1)}\neg end_b \wedge  \widetilde{\Box}_{[l+1, l+2)}\neg beg_b$
\item 
$\varphi_{7}$ : $c \wedge \widetilde{\fut}_{[0,1)} x \Rightarrow 
  (\neg x \widetilde{\until}_{[0,1)}[x \wedge (\neg x \wedge \neg c)
  \until_{(0,1)} (c \vee EP)])$ for $x \in \{end_b,beg_b\}$.
 \end{itemize}
 Note that above formulae uniquely determine the points where $c, end_b,beg_b$
 must hold in $\rho'$ based on where $a$ holds in $\rho'$.
  Using these extra propositions, we now 
construct a formula which enforces the other direction 
$act \Rightarrow (b \Rightarrow \past_{[l,l+1)} (a \wedge act))$ of $\hat{X}_{[l,l+1)}$
within interval  $[t+l+1, t+l+2)$, $t \in \N$.  We sketch this proof case-wise. For 
$t \in \N$, 
 
 \noindent\underline{\bf Case 1}:  If $act \wedge \neg a$ holds throughout 
  $(t, t+2)$,   then, $\past_{[l, l+1)}(a \wedge act)$ cannot hold anywhere in $[t+l+1, t+l+2)$ 
  (if it did, then we will have an $a \wedge act$ in $(t, t+2)$). 

 \noindent \underline{\bf Case 2}: If $\neg a$ holds at all points $act$ in $(t, t+1]$, 
 and if there is an $a \wedge act$ in $(t+1, t+2)$. 
 Assume that the first $a \wedge act$ in $[t+1, t+2)$ occurs at $s=t+ 1+ \epsilon$. Then, by $\varphi_2$, 
 we have a $end_b$ at $s+l=t+1+\epsilon+l$. Also, $\neg beg_b$ holds throughout $(t+l+1, t+l+2)$. 
 $\past_{[l, l+1)}(a \wedge act)$ cannot hold at points $act$ in $[t+l+1, s+l)$, for it did, 
 then we must have an $a \wedge act$ in $(t, s)$. 
  The formula $\varphi_{9}$ considers cases 1 and 2.  
\begin{itemize}
 \item $\varphi_{9}: \widetilde{\Box}[(c\wedge \widetilde{\Box}_{[0,1)} \neg beg_b) \Rightarrow 
  ((act \Rightarrow \neg b) \wedge \neg x) \widetilde{\until} x]$ where $x=(end_b \vee c \vee EP)$.
\end{itemize}
  \noindent \underline{\bf Case 3} If $\neg a$ holds at all points $act$ in $[t+1, t+2)$, 
  and if there is an $act \wedge a$ in $[t, t+1)$.   Assume that the last $a \wedge act$ 
  in $[t, t+1)$ occurred at $t+\delta=v$, $ 0 \leq \delta < 1$. Then 
  by $\varphi_{8}$, we have $b$ holds at all points $act$ in $[v+l, v+l+1)$. 
  Also, by $\varphi_4$, $beg_b$ holds at $v+l+1$, and 
  $\neg end_b$ holds throughout $[t+l+1, t+l+2)$ by $\varphi_{6}$. 
   However, we cannot have a  $b \wedge act$ in $[v+l+1, t+l+2)$, for this would mean the presence of an $a \wedge act$ in   $(v, t+2)$.  Note that if the last $a \wedge act$ of $[t, t+1)$ is at $t$, then $beg_b$ holds at $t+l+1$.
 
  \label{figs}
 \begin{figure}[ht]
 \begin{center}
 \begin{picture}(45,8)(20,-5)
 \thicklines
 \drawline[AHnb=0,ATnb=0](-15,0)(105,0)
 \drawline[AHnb=0,ATnb=0](-15,-2)(-15,2)
 \drawline[AHnb=0,ATnb=0](105,-2)(105,2)
 \put(-16,-4){$t$}
 \put(-15,4){$c$}
 \drawline[AHnb=0,dash={1.5}0](-15,-6.5)(15,-6.5)
 \put(-15,-6.5){(}
 \put(15,-6.5){)}
 \put(0,-8.5){$\neg a$}
 \put(101,-4){$t+l+2$}
 \put(104,4){$c$}
 \drawline[AHnb=0,ATnb=0](0,-2)(0,2)
 \put(-3,-4){$t+1$}
 \put(-1,4){$c$}
 \drawline[AHnb=0,ATnb=0](16,-2)(16,2)
 \put(13,-4){$t+2$}
 \put(15,4){$c$}
 \drawline[AHnb=0,ATnb=0](85,-2)(85,2)
 \put(81,-4){$t+l+1$}
 \put(83.5,4){$c$}
 \drawline[AHnb=0,dash={1.5}0](85,-6)(103.5,-6.5)
 \put(84.5,-6.5){[}
 \put(104,-6.5){)}
 \put(90,-8.5){$\neg b$}
 \drawline[AHnb=0,ATnb=0](65,-2)(65,2)
 \put(61,-4){$t+l$}
 \put(64,4){$c$}
 \end{picture}
 \caption{Case 1}
 \label{case1}
 \end{center}
 \end{figure}

 \begin{figure}[h]
 \begin{center}
 \begin{picture}(45,9)(20,-6)
 \thicklines
 \drawline[AHnb=0,ATnb=0](-15,0)(105,0)
 \drawline[AHnb=0,ATnb=0](-15,-2)(-15,2)
 \drawline[AHnb=0,ATnb=0](105,-2)(105,2)
 \put(-16,-4){$t$}
 \put(-15,3){$c$}
 \drawline[AHnb=0,dash={1.5}0](-15.2,6)(3,6)
 \put(-15.5,6){(}
 \put(3,6){)}
 \put(-7,7.5){$\neg a$}
 \put(101,-4){$t+l+2$}
 \put(104,3){$c$}
 \drawline[AHnb=0,ATnb=0](0,-2)(0,2)
 \put(-5,-4){$t+1$}
 \put(0,3){$c$}
 \drawline[AHnb=0,ATnb=0](5,-2)(5,2)
 \put(4,-4){$s$}
 \put(5,5){$a$}
 \drawline[AHnb=0,ATnb=0](16,-2)(16,2)
 \put(13,-4){$t+2$}
 \put(15,3){$c$}

 \drawline[AHnb=0,ATnb=0](80,-2)(80,2)
 \put(77,-4){$s+l$}
 \put(79,3){$end_b$}

 \drawline[AHnb=0,dash={1.5}0](65,-7)(79,-7)
 \put(65,-7){[}
 \put(79,-7){)}
 \put(73,-10){$\neg b$}
 
 \drawline[AHnb=0,dash={1.5}0](65,7)(104,7)
 \put(65,7){(}
 \put(104,7){)}
 \put(90,10){$\neg beg_b$}

 \drawline[AHnb=0,ATnb=0](65,-2)(65,2)
 \put(61,-4){$t+l+1$}
 \put(64,3){$c$}
 \end{picture}
 \caption{Case 2}
 \label{case2}
 \end{center}
 \end{figure}

 \begin{figure}[!h]
 \begin{center}
 \begin{picture}(45,8)(20,-7)
 \thicklines
 \drawline[AHnb=0,ATnb=0](-15,0)(105,0)
 \drawline[AHnb=0,ATnb=0](-15,-2)(-15,2)
 \drawline[AHnb=0,ATnb=0](105,-2)(105,2)
 \put(-16,-4){$t$}
 \put(-15,3){$c$}
 \drawline[AHnb=0,dash={1.5}0](-2,6)(15,6)
 \put(-2,6){(}
 \put(15,6){)}
 \put(5,7.5){$\neg a$}
 \put(101,-4){$t+l+2$}
 \put(105,3){$c$}
 \drawline[AHnb=0,ATnb=0](0,-2)(0,2)
 \put(-1,-4){$t+1$}
 \put(-1,3){$c$}
 \drawline[AHnb=0,ATnb=0](-5,-2)(-5,2)
 \put(-6,-4){$v$}
 \put(-6,5){$a$}
 
 \drawline[AHnb=0,ATnb=0](16,-2)(16,2)
 \put(13,-4){$t+2$}
 \put(15,3){$c$}

 \drawline[AHnb=0,ATnb=0](71,-2)(71,2)
 \put(69,-4){$t+l+1$}
 \put(70.5,3){$c$}
 \drawline[AHnb=0,ATnb=0](93,-2)(93,2)
 \put(85,-4){$v+l+1$}
 \put(92,3){$beg_b$}
 
 \drawline[AHnb=0,dash={1.5}0](71.5,-7)(105,-7)
 \put(71,-7){[}
 \put(105,-7){)}
 \put(80,-10){$\neg end_b$}

 \drawline[AHnb=0,dash={1.5}0](93.5,6)(103,6)
 \put(93,6){[}
 \put(103,6){)}
 \put(96,8){$\neg b$}

 \drawline[AHnb=0,ATnb=0](50,-2)(50,2)
 \put(45,-4){$t+l$}
 \put(49,3){$c$}
 \drawline[AHnb=0,ATnb=0](63,-2)(63,2)
 \put(58,-4){$v+l$}
 
 \end{picture}
 \caption{Case 3}
 \label{case3}
 \end{center}
 \end{figure}

 \begin{figure}[!h]
 \begin{center}
 \begin{picture}(45,9)(20,-8)
 \thicklines
 \drawline[AHnb=0,ATnb=0](-15,0)(110,0)
 \drawline[AHnb=0,ATnb=0](-15,-2)(-15,2)
 \drawline[AHnb=0,ATnb=0](105,-2)(105,2)
 \put(-16,-4){$t$}
 \put(-15,3){$c$}
 
 \drawline[AHnb=0,dash={1.5}0](-6,-7)(6,-7)
 \put(-6,-7){(}
 \put(6,-7){)}
 \put(0,-10){$\neg a$}
 \put(99,-4){$t+l+2$}
 \put(104,3){$c$}
 
 \drawline[AHnb=0,ATnb=0](0,-2)(0,2)
 \put(-2,-4){$t+1$}
 \put(-1,3){$c$}
 
 \drawline[AHnb=0,ATnb=0](-5,-2)(-5,2)
 \put(-6,-4){$u$}
 \put(-6,3){$a$}

 \drawline[AHnb=0,ATnb=0](7,-2)(7,2)
 \put(6,-4){$v$}
 \put(6,3){$a$}
 \drawline[AHnb=0,ATnb=0](15,-2)(15,2)
 \put(12,-4){$t+2$}
 \put(14,3){$c$}

 \drawline[AHnb=0,ATnb=0](60,-2)(60,2)
 \put(55,-4){$t+l+1$}
 \put(59,3){$c$}
 
 \drawline[AHnb=0,ATnb=0](75,-2)(75,2)
 \put(73,-4){$u+l+1$}
 \put(74,3){$beg_b$}
 
 \drawline[AHnb=0,ATnb=0](94,-2)(94,2)
 \put(89,-4){$v+l$}
 \put(93,3){$end_b$}

 \drawline[AHnb=0,ATnb=0](40,-2)(40,2)
 \put(35,-4){$t+l$}
 \put(39,3){$c$}
 
 \drawline[AHnb=0,ATnb=0](50,-2)(50,2)
 \put(45,-4){$u+l$}
 \drawline[AHnb=0,dash={1.5}0](74.8,6)(92,6)
 \put(74.5,6){[}
 \put(92.6,6){)}
 \put(83,9){$\neg b$}
 \end{picture}
 \caption{Case 4}
 \label{case4}
 \end{center}
 \end{figure}

 \noindent\underline{\bf Case 4} If there is an $a \wedge act$ in both $[t+1, t+2)$ and $[t, t+1)$.
 Assume that the last $a \wedge act$ in $[t, t+1)$ is at $u=t+\epsilon$, and the first $a \wedge act$ 
 in $[t+1, t+2)$ is at $v=t+1+\kappa$, with $\epsilon, \kappa \geq 0$. 
 If $v-u \leq 1$, then
 $\epsilon \geq \kappa$, and by $\varphi_{8}$, we have $b$ holds at all points $act$ in $[t+\epsilon+l, t+l+2+\kappa)$. 
  However, if $v-u > 1$, then $\kappa > \epsilon$, and by $\varphi_{8}$, 
 $b$ holds at all points $act$ of $[u+l, u+l+1)$ and $[v+l, v+l+1)$, 
 with $u+l+1 < v+l$, In this case, all points $act$ 
 in the range $[u+l+1, v+l)$ must be marked $\neg b$. 
The following formula handles cases 3 and 4. For $x=\neg(end_b \vee c \vee EP)$,
\begin{itemize}
 \item $\varphi_{10}: \widetilde{\Box}\{(c \wedge [\neg end_b
\widetilde{\until}_{[0,1)} beg_b])
\Rightarrow \widetilde{\fut}_{[0,1)}[(beg_b \wedge 
(((act \Rightarrow \neg b) \wedge x)\widetilde{\until} (end_b \vee c \vee EP)))] \}$ 
\end{itemize}
Let $\psi= \bigwedge_{i=1}^{10}\varphi_i \in \mtl$.  

 \noindent {\bf Proof of 1.} We claim that 
 $\rho',i \models \psi$ implies $\rho',i \models \hat{X}_{[l,l+1)}$. 
   Assume that $\rho',i \models \psi$.  Let $t_i \in [t+l, t+l+1)$ for some $t \in \N$.  
  Suppose $\rho', i \nvDash \past_{[l, l+1)}(a \wedge act)$ and $\rho', i \models act$.  
  We show that $\rho', i\models \neg b$.
  
  Since $\rho', i \nvDash \past_{[l, l+1)}(a \wedge act)$, 
  all points  $act$ in $(t_i-l-1, t_i-l]$ are marked $\neg a$. 
 Note that $(t_i-l-1, t_i-l] \subset [t-1,t+1)$ with $t-1 \leq t_i-l-1 < t$, and 
 $t \leq t_i-l < t+1$. 
 \begin{enumerate}
  \item We have $\widetilde{\Box}(act \Rightarrow \neg a)$ in $(t_i-l-1, t_i-l]$. 
  Assume that there is an $a \wedge act$ in $[t-1,t)$, and 
  the last such occurs at $u \leq t_i-l-1$. Assume further 
  that there is an $a \wedge act$ in $[t,t+1)$, and the first such 
   occurs at $v > t_i-l$. 
  Then, by case 4 of the analysis, we obtain $\Box \neg b$ at all points $act$ 
  in $[u+l+1, v+l)$. Clearly, $u+l+1 \leq t_i < v+l$, hence $act \wedge \neg b$ 
  holds at $t_i$.
 \item Assume that there is no $a \wedge act$ in $[t-1,t)$, 
 but there is an $a \wedge act$ in $[t,t+1)$. 
 The first such   $a \wedge act$ occurs at $s > t_i-l$. 
  Then, by case 2 of our analysis, $\neg b$ holds at all points $act$ 
  in $[t+l,s+l)$. Clearly, $t+l \leq t_i < s+l$, hence, 
  $act \wedge \neg b$   holds at $t_i$.
  \item Assume that there is an $a \wedge act$ in $[t-1,t)$, and 
  the last such occurs at $v \leq t_i-l-1$. Further, assume there 
  is no $a \wedge act$ in $[t,t+1)$. Then, by case 3 
  of our analysis, $\neg b$ holds at all points $act$ of 
  $[v+l+1, t+l+1)$. Clearly, $v+l+1 \leq t_i < t+l+1$. 
  hence $act \wedge \neg b$ 
  holds at $t_i$.
  \item Assume that there is no $a \wedge act$ 
  in both $[t-1,t)$ as well as $[t,t+1)$. 
  In this case, by case 1, $act \wedge \neg b$ holds at all points 
  of $[t+l, t+l+1)$. Clearly, $t_i \in [t+l, t+l+1)$, 
  hence $act \wedge \neg b$   holds at $t_i$.
   \end{enumerate}
  Thus, $\rho', i \models \neg b$, and hence $\rho', i \models (act \Rightarrow 
 (\neg \past_{[l,l+1)}(a \wedge act) \Rightarrow \neg b))$. 
  
Now assume that $\rho', i \models act \wedge \neg b$. 
 We show that $\rho', i \models \neg \past_{[l,l+1)}(a \wedge act)$. 
 Suppose $\rho', i \models \past_{[l,l+1)}(a \wedge act)$. Then there is a point 
 $t \in (t_i-l-1, t_i-l]$ where $a \wedge act$ holds. 
 Then, by $\varphi_{8}$, we have $(act \Rightarrow b)$ holds 
 at all points of $[t+l, t+l+1)$. Note that $t_i \in [t+l, t+l+1)$, and henceforth
 $\rho', i \models  (act \Rightarrow b)$, which contradicts the assumption 
 we started out with. Hence, $\rho', i \models \neg \past_{[l,l+1)}(a \wedge act)$.

 Now consider the case of a point $act$ at $t_i \in [0, l)$. Clearly, for such a $t_i$, 
 $\past_{[l, l+1)}(a \wedge act)$ cannot hold.
$\varphi_2-\varphi_7$ assert that (i) there is no $end_b$ in  $[0,l)$ and there is no $beg_b$ in $[0, l+1)$, 
(ii) if in some unit interval with integral end points, 
there is no $end_b$ and $beg_b$, then in that interval all points $act$ will be marked $ \neg b$.  
Thus, in $[0,l)$ all points $act$ are marked $\neg b$.  
 At timestamps $t \ge l$, all points $act$ satisfying $\past_{[l, l+1)}(a \wedge act)$ are marked $b$
 by $\varphi_{8}$. We have thus showed $\rho',i \models \psi$ implies $\rho',i \models \hat{X}_{[l,l+1)}$. 

\noindent {\bf Proof of 2.} Assume that
 $\rho' \models \hat{X}_{[l,l+1)}$. Then 
 we can construct $\rho'' \in T\Sigma''^*$ such that 
$\rho'' \models \psi$, and $\rho'' \Downarrow^{\Sigma'}_{\Sigma''}= \rho' \Downarrow^{\Sigma'}_{\Sigma''}$.
Assume that for any $\rho' \in T\Sigma''^*$, $\rho' \models \hat{X}_{[l,l+1)}$.
Then, at any point $i$ of $\rho'$, $\rho', i \models act$ iff 
$\rho', i \models (b \Leftrightarrow (\past_{[l,l+1)}(a \wedge act))$. 

Consider the word $\hat{\rho}=\rho'\Downarrow^{\Sigma'}_{\Sigma''}$. 
$\rho' \models \hat{X}_{[l,l+1)}$ implies $\hat{\rho} \models \hat{X}_{[l,l+1)}$.
From $\hat{\rho}$, construct as given by the formulae 
$\varphi_1$ to $\varphi_7$ of Lemma \ref{past-b1}, the oversampling $\rho'' \in T\Sigma''^*$. 
That is, $\rho'' \models \bigwedge_{i=1}^7 \varphi_i \wedge \hat{X}_{[l,l+1)}$.
Now, we first show that $\rho'' \models \psi$. If not, 
then $\rho'' \models \neg \varphi_8 \vee \neg \varphi_9 \vee \neg \varphi_{10}$. 
\begin{enumerate}
 \item[(a)] Assume $\rho'' \models \neg \varphi_8$. Then, 
 there exists $i$ such that $\rho'',i \models (a \wedge act)$, and 
 $\rho'',i \models  \fut_{[l,l+1)}(act \wedge \neg b)$. 
 Let $t_j \in [t_i+l, t_i+l+1)$ be the point where $act \wedge \neg b$ holds.
 Then, we have $\rho'',j \models (act \wedge \neg b) \wedge \past_{[l,l+1)}(a \wedge act)$
 (recall that $t_i \in (t_j-l-1, t_j-l]$, and $\rho'',i \models (a \wedge act)$). That is, $\rho'',j \models act$ and 
 $\rho'',j \models (\neg b \wedge \past_{[l,l+1)}(a \wedge act))$. Hence, 
 $\rho''\models \neg \hat{X}_{[l,l+1)}$, a contradiction. 
 \item[(b)] Assume $\rho'' \models \neg \varphi_9$. Then, 
 there exists an integral point $i$ such that 
 $\rho'', i \models (c \wedge \Box_{(0,1)} \neg beg_b)$ and 
 $\rho'', i \nvDash [((act \Rightarrow \neg b) \wedge \neg x) \widetilde{\until} x]$, for 
 $x=(end_b \vee c \vee EP)$. 
 Then, there exists a $j \geq i$, such that $\rho'', j \models (end_b \vee c \vee EP)$
 and there exists $i < k < j$ such that $\rho'', k \models \neg 
 ((act \Rightarrow \neg b) \wedge \neg x)$. Consider the $j$ nearest to $i$ (first 
 point after $i$)  where $(end_b \vee c \vee EP)$ holds. Then, 
 $\rho'', k \models \neg ((act \Rightarrow \neg b) \wedge \neg x)$ for some $i < k <j$ 
 holds when $\rho'', k \models \neg (act \Rightarrow \neg b) \wedge \neg x$. That is, 
 $\rho'', k \models (act \wedge b) \wedge \neg x$.
 
 So, we have now $\rho'', i \models (c \wedge \Box_{(0,1)} \neg beg_b)$ and 
 $\rho'', k \models (act \wedge b) \wedge \neg x$. By $\varphi_4, \varphi_6$, 
 $\rho'', i \models (c \wedge \Box_{(0,1)} \neg beg_b)$ 
 implies that there is no $(a \wedge act)$ in the interval 
 $(t_i-l-1,t_i-l]$, where $t_i$ is the time stamp of $i$. 
 \begin{itemize}
  \item Assume that $\rho'', j \models end_b$. Then, 
 $t_j - t_i < 1$, and by $\varphi_2$, there exists an $a \wedge act$ at 
 $t_j-l$, and that is the first $a \wedge act$ in 
 the unit interval $[t_i-l,t_i+1-l]$. 
 Since $\neg x$ holds at $k$, and $c$ holds at $i$, we have $t_k -t_i < 1$. Also, we 
 have $t_i-l-1 < t_k-l-1 < t_i-l < t_k-l < t_j-l$, and 
 we know that there is no $a \wedge act$ in 
 $(t_i-l-1,t_i-l]$, and the first $a \wedge act$ of 
 $[t_i-l,t_i-l+1]$ occurs at $t_j-l$. Thus, there is 
 no $a \wedge act$ in $(
 t_k-l-1, t_k-l]$. So, we have 
 $\rho'', k \models [\neg \past_{[l,l+1)}(a \wedge act)] \wedge [act \wedge b]$,
 which means $\rho'', k \models \neg \hat{X}_{[l,l+1)}$, a contradiction.
 
\item Assume  that $\rho'', j \models c$. Then $t_j=t_i+1$, and 
there is no $end_b$ in $[t_i,t_i+1]$. Then, by 
$\varphi_2, \varphi_6$, there is no $a\wedge act$ in 
$[t_i-l, t_i-l+1)$. Then, in this case, there is no $a \wedge act$ 
in $(t_i-l-1,t_i-l)$ and $[t_i-l,t_i-l+1)$.  Since 
$(t_k-l-1, t_k-l] \subseteq (t_i-l-1, t_i-l+1)$, we have 
 $\rho'', k \models [\neg \past_{[l,l+1)}(a \wedge act)] \wedge [act \wedge b]$,
giving $\rho'', k \models \neg \hat{X}_{[l,l+1)}$, a contradiction.

\item Assume  that $\rho'', j \models EP$. Then $t_j-t_i<1$, 
and we have both $\neg beg_b, \neg end_b$ holding 
in $[t_i, t_j]$. Similar to the above case, we can show that 
there is no $a \wedge act$ in $(t_i-l-1,t_i-l]$ and $[t_i-l,t_i-l+1)$, and hence 
arrive at the same contradiction.
\end{itemize}
 \item[(c)] Assume $\rho'' \models \neg \varphi_{10}$. This case is similar to the case 
 when $\rho'' \models \neg \varphi_{9}$.
 \end{enumerate}
So we have proved that $\rho'' \models \psi$. 
Recall that $\hat{\rho}=\rho'\Downarrow^{\Sigma'}_{\Sigma''}$, and 
$\rho''$ was constructed by adding oversampling points to 
$\hat{\rho}$. Hence, $\rho'' \Downarrow^{\Sigma'}_{\Sigma''}=\hat{\rho}=\rho'\Downarrow^{\Sigma'}_{\Sigma''}$, 
giving the proof.
\end{proof}

\begin{theorem}
\label{theo:main}
For every $\varphi \in \mtlsns$ over $\Sigma$, we can construct $\psi_{fut} \in \mtl$ 
over $\Sigma'' \supseteq \Sigma$ such that
\begin{enumerate}
\item For all $\rho' \in T\Sigma''^*$, if $\rho' \models \psi$ then $\rho' \Downarrow^\Sigma_{\Sigma''} \models \varphi$.
\item For all $\rho \in T\Sigma^*$, if $\rho \models \varphi$ then there exists
$\rho' \in T\Sigma''^*$ such that $\rho' \models \psi$ and $\rho' \Downarrow^\Sigma_{\Sigma''} = \rho$.
\end{enumerate}
\end{theorem}
\begin{proof}
Note that $\since_{NS}$ can be expressed using $\since$ and 
$\past_{NS}$ \cite{deepak08}. For instance, we can write $a \since_{[l, r)}b$ as 
$\past_{[l, r)}b ~\wedge ~(a \since b) \land \boxminus_{[0,l)}(a \wedge a \since b)$, for $r=l+1, \infty$. 
Similarly,  all intervals $\langle l, l+1 \rangle$, $\langle l, \infty \rangle$
are handled.  Further, $\since$ can be 
removed (More details can be found at Appendix \ref{since-rem}) \cite{deepak08}, \cite{formats11} to obtain an equisatisfiable $\mtl$ formula.
Also, $\past_{[l,m)}$ is equivalent to 
$\past_{[l,l+1)} \lor  \past_{[l+1,l+2)} \lor \cdots \lor \past_{[m-1,m)}$. Hence,
the only past modalities in the formulae are $\past_{[l,\infty)}$ or $\past_{[l,l+1)}$.
Lemmas \ref{remove-pastinf} and \ref{past-b1} show how these can be expressed
in $\mtl$ to obtain equisatisfiable formulae. Hence the theorem follows.
\end{proof}
By symmetry, using reflection \cite{formats11}, 
we can reduce  $\mtluns$ to $\mtlsns$.  
Appendix \ref{ex} illustrates in detail, the elimination of a past modality $\past_{[l,l+1)}a$.
\subsection{Expressiveness}
We wind up this section with a brief discussion 
about the expressive powers of logics 
$\mtlsns$ and $\mtluns$. The following lemma highlights 
the fact that even unary modalities $\fut_I ,\past_I$ with singular intervals 
are more expressive than $\until_{NS}, \since_{NS}$; likewise, 
non-singular intervals are more expressive than intervals 
of the form $[0, \infty)$.

\begin{lemma}
\label{game:proof}
(i) $\mtlfutpw \nsubseteq \mtluns$, 
 (ii) $\mtlfutp \nsubseteq \mtlsns$, and (iii) $\mitlfp \nsubseteq \mtlu$.
\end{lemma}
\begin{proof}
The formula $\fut_{(0,1)}\{a \wedge \neg \fut_{[1,1]}(a \vee b)\}$
in $\mtlfutpw$ has no equivalent formula in $\mtluns$. 
Similarly, the formula  $\fut\{b \wedge  \neg \past_{[1,1]}(a \vee b)\}$
in $\mtlfutp$ has no equivalent formula in $\mtlsns$. 
The formula $\fut_{(1,2)}[a \wedge \neg \past_{(1,2)}a] \in \mitlfp$ 
has no equivalent formula in $\mtlu$. 
A proof using EF games \cite{concur11}
can be seen in Appendix \ref{game:sec}. 
\end{proof}

\section{Unary MTL and Undecidability}
We explore the unary fragment of $\mathsf{MTL}$.
In this section, we show the undecidability of satisfiability checking of 
$\mtlfullunary$ over finite timed words. The undecidability follows 
by construction of an appropriate $\mathsf{MTL}$ formula $\varphi$ simulating a deterministic k-counter counter machine ${\cal M}$ such that 
$\varphi$ is satisfiable iff ${\cal M}$ halts. We also show the non primitive recursive lower bound for satisfiability of $\mtlfutpw$ by reduction of halting problem (location reachability problem) for counter machine with increment errors 
\cite{schnobelen02}, \cite{demriL06} to satisfiability of the logic.

 A deterministic k-counter machine is a  k+1 tuple ${\cal M} = (P,C_1,\ldots,C_k)$, where 
 (i) $C_1,\ldots,C_k$ are k-counters taking values in  $\N$ (their initial values  are set to zero);
and (ii) $P$ is a finite set of instructions with labels $p_1, \dots, p_{n-1},p_n$. 
There is a unique instruction labeled HALT. For $E \in \{C_1,\ldots,C_k\}$, the instructions $P$ are of the following forms: 
 (I) $p_i$: $Inc(E)$, goto $p_j$, (II) $p_i$: If $E =0$, goto $p_j$, else go to $p_k$, (III) $p_i$: $Dec(E)$, goto $p_j$,
 and (IV) $p_n$: HALT. A configuration $W=(i,c_1,\ldots,c_k)$ of ${\cal M}$ at any point of time is given by the value of the current program counter $i$ and valuation of the counters $c_1,\ldots,c_k$. A move of (error-free) counter machine $(l,c_1,\dots,c_k) \rightarrow_{std} (l',c_1',\dots,c_k')$ denotes that configuration $(l',c_1',\dots,c_k')$ is obtained from $(l,c_1,\dots,c_k)$ by executing $l^{th} $ instruction. Subscript $std$ denotes that the
move is that of error-free counter machine. Let $(l^1,c^1_1,\ldots,c^1_k) \leq (l^2,c^2_1,\ldots,c^2_k)$ iff $l^1=l^2$ and $\forall i \in \{ 1,\ldots,k \}$, $c^1_i \leq c^2_i$ . We define a move of a {\em counter machine with increment-errors}
 $(l,c_1,\ldots,c_k) \rightarrow_{incerr} (l'',c_1'',\ldots,c_k'')$ iff $(l,c,d) \rightarrow_{std} (l',c_1',\ldots,c_k')$ and $(l',c_1',\ldots,c_k')\leq (l'',c_1'',\ldots,c_k'')$. Thus, machine may make increment error while moving to a next configuration.
 
 A counter machine whose execution follows the standard moves is called {\em Minsky Counter Machine}. 
 A counter machine whose execution follows moves with increment errors is called {\em Incrementing Counter Machine}. A computation of a counter machine
 (of given type)
 is a sequence of moves (of appropriate type) $W_0 \rightarrow W_1 \ldots \rightarrow W_m$ where
 $W_0 =(1,0,\ldots,0)$. The computation is terminating if the last configuration is a halting
 configuration, i.e. $C_m=(n,c^m_1,\ldots,c^m_k)$. A counter machine is called halting if it has a
 terminating computation.
 
 \oomit{
A computation of ${\cal M}$ is a finite sequence of configurations $C_0 C_1 \dots C_i C_{i+1} \dots$ 
where $C_j=\langle p, c, d\rangle$ is obtained from $C_{j-1}=\langle q, c', d'\rangle$ 
 by execution of the instruction indicated by $q$. The initial configuration is 
$\langle 1, 0, 0 \rangle$, and ${\cal M}$ is said to halt if the last instruction executed is the 
HALT instruction. Given ${\cal M}$, it is undecidable whether ${\cal M}$ halts \cite{minsky}.
}

\begin{theorem}[\cite{minsky}] 
\label{theo:minsky}
Whether a given $k$-counter ($k \geq 2$) Minsky machine is halting or not (equivalently
the location reachability problem) is undecidable.
\end{theorem}

\begin{theorem}[\cite{schnobelen02,demriL06}]
\label{theo:schnobelen}
 Whether a given $k$-counter incrementing machine is halting or not (equivalently
the location reachability problem) is decidable with non primitive recursive complexity.
\end{theorem}

\subsection*{Encoding Minsky Machines in $\mtlfullunary$}
\label{minsky}
We encode each computation of a k-counter machine ${\cal M}$ using 
(a non-empty set of equivalent) timed words over the alphabet $\Sigma_{\cal M}=\{b_1, b_2, \dots, b_n, a\}$. 
The timed language $L_ {\cal M}$ over $\Sigma_{\cal M}$ contains one and only one timed word 
corresponding to unique halting computation of  ${\cal M}$. We then generate a 
formula $\varphi_{\cal M}$  such that $L_ {\cal M}=L(\varphi_{\cal M})$.  
The encoding is done in the following way: A configuration $\langle i,c_1,\ldots,c_k \rangle$  is represented by a sub-string with untimed part 
$b_i a^{c_1} a^{c_2} \ldots a^{c_k}$. A computation of ${\cal M}$ is encoded by concatenating sequences of 
individual configurations. We encode the $j^{th}$ configuration of ${\cal M}$ in the 
time interval $[(2k+1)j,(2k+1)(j+1))$ as follows: 
For $j \in \N$,\\
(i) $b_{i_{j}}$ (representing instruction $p_{i_j}$) occurs at time $(2k+1)j$; 
(ii) The value of counter $C_q$, $q \in \{1,2,\dots,k\}$, in the $j^{th}$ configuration is given by the number of $a$'s 
in the interval $((2k+1)j+2q-1,(2k+1)j+2q)$; (iii) The $a$'s can appear only in the intervals 
$((2k+1)j+2q-1,(2k+1)j+2q)$, $q \in \{1,2,\dots,k\}$, and (iv) The intervals 
$((2k+1)j+2w,(2k+1)j+2w+1)$, $w \in\{0,\ldots,k\}$ have no events.\\
The computation must start with initial configuration and the final configuration must be
the $HALT$ instruction; beyond this, there are no more instructions.
\oomit{
A timed word $\rho=(\sigma,\tau)$ is in $L_{\cal M}$ iff it follows the above specifications.
We now proceed to give the formula $\varphi_{\cal M}$ over $\Sigma_{\cal M}$ which describes $L_{\cal M}$.
}  
$\varphi_{\cal M}$ is obtained as a conjunction of several formulae. Let $B$ be a shorthand for $\bigvee_{i\in {1,\ldots,n}}b_{i}$. 
We first give some generic formulae which hold for both Minsky and Incrementing machines.\\
\noindent 1. The symbol $b_{i_j}$ representing instruction $p_{i_j}$ occurs at $(2k+1)j$ for all $j \in \N$:\\
$
~~~~~~~~~~~~~~~~~\varphi_{0}\ =\ b_1\ \wedge \widetilde{\Box}\{(B \wedge \fut B) \Rightarrow \fut_{[2k+1,2k+1]}B\} \wedge \widetilde{\Box}\{B \Rightarrow (\neg \fut_{(0,2k+1)} B) \}.  \\
$
\noindent 2. No events in intervals $((2k+1)j+2w,(2k+1)j+2w+1)$, $w \in \{0,\dots,k\}$, $j \in \N$.\\
$
~~~~~~~ \varphi_{1}\ =\ \widetilde{\Box}\{B\ \Rightarrow\ \bigwedge_{w \in \{0,\ldots,k\}} (\Box_{[2w,2w+1]}(\neg a)). 
$\\
\noindent 3. Beyond $p_n$=HALT, there are no instructions: $\varphi_{2}\ =\  \widetilde{\Box}(b_n \Rightarrow 
\Box_{[2k+1, \infty)} false)\\
$
\noindent 4. Computation starts in $(1,0,\ldots,0)$ : $\varphi_4\ =\ b_1 \wedge \Box_{(0,2k+1)}false\\
$
\noindent 5. At any point of time, exactly one event takes place. Events have distinct time stamps.\\
$
 ~~~~~~~~~~~~~\varphi_5\ =\ [\bigwedge_{y \in \Sigma_{\cal M}}(y \Rightarrow
  \neg(\bigvee_{x \in \Sigma_{\cal M} \setminus \{y\}}(x)) ]\\ 
$
\noindent 6. Eventually we reach the halting configuration $\langle p_n,c_1,\ldots,c_k \rangle$: $\varphi_6 = \widetilde{\fut} b_n\\
$
\noindent 7. We define macros, $COPY_i$,$INC_i$,$DEC_i$ for counter $C_i$.

\begin{itemize}
 \item $COPY_i$: Every $a$ occurring in the interval $((2k+1)j+2i-1,(2k+1)j+2i)$ has a copy at a future distance $2k+1$, 
and every $a$ occurring in the next interval has an $a$ at a past distance $2k+1$. This ensures the absence 
of insertion errors.\\
   $COPY_i = \Box_{(2i-1,2i)}[(a \Rightarrow \fut_{[2k+1,2k+1]} a)] \wedge \Box_{((2k+1)+2i-1,(2k+1)+2i)}[(a\Rightarrow \past_{[2k+1,2k+1]} a)]$. 
 
\item $INC_i$: All $a$'s in the current configuration are copied to the next, 
at a future distance $2k+1$; every $a$ except the last, in the next configuration 
has an $a$ at past distance $2k+1$.
  \begin{quote}  
 $INC_i = 
 \begin{array}[t]{l}  
   (\Box_{(2i-1,2i)}\{ ~[a \Rightarrow \fut_{[2k+1,2k+1]} a] \\
   \wedge [(a \wedge \neg \fut_{(0,1)}a) \Rightarrow (\fut_{(2k+1,2k+2)}a 
   \wedge \Box_{(2k+1,2k+2)}(a \Rightarrow \Box_{(0,1)}(false)))] ~\}\\
     ~~~~~~~~ 
   \wedge \Box_{((2k+1)+2i-1,(2k+1)+2i)} [(a \wedge \fut_{[0,1]}(a)) \Rightarrow \past_{[2k+1,2k+1]}a])
 \end{array}
 $
  \end{quote}

\item $DEC_i$: All the $a$'s in the current configuration,  except the last,
have a copy at future distance $2k+1$. All the $a$'s in the next configuration have a copy at past distance $2k+1$.
\begin{quote}
$DEC_i = \Box_{(2i-1,2i)}\{ [(a \wedge \fut_{(0,1)}a)\Rightarrow \fut_{[2k+1,2k+1]} a] 
\wedge [(a \wedge \neg \fut_{[0,1]}a) \Rightarrow \neg \fut_{[2k+1,2k+2]} a]\} \wedge 
\Box_{((2k+1)+2i-1,(2k+1)+2i)}[(a\Rightarrow \past_{[2k+1,2k+1]} a)]$.  
\end{quote}
\end{itemize}
These macros can be used to simulate all type of instructions. 
We explain only the zero-check instruction here.   $p_x$: If $C_i=0$ goto $p_y$, else goto $p_z$.\\
$ \varphi^{x,i=0}_{3}\ =\ \widetilde{\Box}\{b_x \Rightarrow (\bigwedge_{i\in \{1,\ldots,n\}} COPY_i \wedge 
[\Box_{(2i-1,2i)}(\neg a) \Rightarrow ( \fut_{[2k+1,2k+1]}b_y)]\wedge\\
~~~~~~~~~~~~~~~~~~~~~[\fut_{(2i-1,2i)}(a) \Rightarrow (\fut_{[2k+1,2k+1]}b_{z})]\}$\\
The encodings $\varphi^{x,inc_i}_{3}, \varphi^{x,dec_i}_{3}$, corresponding to increment, decrement instructions of counter $i$ can be found in Appendix \ref{undec2:sec}.
The final formula we construct is $\varphi_{\cal M} = \bigwedge_{i=0}^6 \varphi_i$,  where $\varphi_3$ is the conjunction 
of formulae $\varphi^{x,inc_i}_{3}, \varphi^{x,dec_i}_{3}, \varphi^{x,i=0}_{3}$, $i \in \{1,2, \dots,k\}$. 

\subsection*{Encoding Incrementing Counter Machines in $\mtlfutpw$}
\label{nprdec}
To encode a computation of incrementing counter machine ${\cal M}$, we need to represent increment, decrement and no change of counter values in presence of increment errors. Since increment errors need not be checked, the encoding does not need past modality : all formulae except $COPY_i,INC_i,DEC_i$ are in $\mtlfutpw$.
We now give $COPYERR_i,INCERR_i,DECERR_i$ in place of $COPY_i,INC_i,DEC_i$ which allows insertion errors.\\
\noindent 1.  Copy counter with error: Copy all $a$'s without restricting insertions of other $a$'s.\\
$COPYERR_i = \Box_{(2i-1,2i)}[ (a\Rightarrow \fut_{[2k+1,2k+1]} a)]$\\
\noindent 2. Increment counter with insertion errors: Copy all $a$'s inserting at least 1 $a$ after the last copied $a$. $ INCERR_i$ is defined as $\Box_{(2i-1,2i)} \{ [(a \wedge \fut_{(0,1)}a)\Rightarrow (\fut_{[2k+1,2k+1]} a)] $\\
$\wedge [(a \wedge \neg \fut_{(0,1]}a)\Rightarrow (\fut_{[2k+1,2k+1]} a \wedge \fut_{(0,1)}a) ] \}$.\\
 \noindent 3. Decrement counter with error: Copy all the $a$'s in an interval, except the last $a$ \\
$  DECERR_i = 
  \Box_{(2i-1,2i)}  [(a \wedge \fut_{(0,1)}a)\Rightarrow \fut_{[2k+1,2k+1]} a]\\
$
Construction of $\varphi_{\cal M}$ corresponding to the incrementing counter machine 
is done using these macros. This is similar to the construction 
in Appendix \ref{undec2:sec}.

\begin{lemma}
 \label{nprdec:proof}
Let ${\cal M}$ be a $k$-counter incrementing machine. Then, we 
can synthesize a formula $\varphi_{\cal M} \in \mtlfutpw$
such that  ${\cal M}$ halts iff  $\varphi_{\cal M}$ is satisfiable.
\end{lemma}

\begin{lemma}
 \label{undec2:proof}
Let ${\cal M}$ be a $k$-counter Minsky machine. 
Then, we 
can synthesize a formula $\varphi_{\cal M} \in \mtlfullunary$
such that  ${\cal M}$ halts iff  $\varphi_{\cal M}$ is satisfiable.
\end{lemma}
Lemma \ref{nprdec:proof} and Theorem \ref{theo:schnobelen} together say that satisfiability
of  $\mtlfutpw$ is non primitive recursive.
Lemma\ref{undec2:proof} and  Theorem \ref{theo:minsky} together say that satisfiability
of  $\mtlfullunary$ is undecidable. 
It follows from  lemma\ref{undec2:proof} that the recurrent state problem of k-counter incremental machine 
can also be encoded; hence $\mtlfutpw$ is undecidable over infinite words.

\section{Discussion}
We have shown that satisfiability of $\mtlsns$ over finite strictly monotonic timed words is decidable. This subsumes the  previously known decidable fragments $\mitl$ and $\mtl$. The decidability  proof is carried out by extending the technique of temporal projections 
\cite{fundinfo04,deepak08,PrabhakarD06,formats11} to pointwise models in
presence of oversampling. In general, this technique allows us to reduce a formula
of one logic to an 
 equi-satisfiable formulae in a different/simpler logic. 
  We believe that the technique of temporal projections with oversampling has wide applicability and it embodies an interesting notion of 
 equivalence of formulae/logics modulo temporal projections.

In the second part of the paper, we have investigated the decidability of the unary fragment $\gmtlfullunary$ 
which is expressively weaker than full $\gmtlfull$ \cite{concur11}. 
As observed by Rabinovich, the standard construction encoding a  $k$-counter machine configuration in 
unit interval does not work in absence of $\until$ (or $\since$ operator). We have  
arrived at an altered encoding of a configuration using a time interval of length $2k+1$ with suitable gaps. We have shown that the restriction of $\gmtlfull$ to its unary fragment does not lead to any improvement in decidability. Using similar ideas, perfect channel machines can also be encoded into $\mtlfullunary$ and lossy channel machines can be
encoded into $\mtlfutpw$. Our exploration has mainly looked at pointwise models with strictly monotonic time. The case of weakly monotonic time requires more work.

\newpage
\appendix
\centerline{\Large \bf Appendix}

\section{Proof of Lemma \ref{lem:gen}}
\label{proof:gen}
The proof idea is to apply structural induction on $\hat{\varphi}$ taking care of the non action points
that can get added to $\rho'$. 
The base case involves formulae of the form $a \until_I b$.
Let $\rho=(X_1,0)(X_2,t_2)\dots (X_k,t_k)$. If $\rho,i \models a \until_I b$, 
then there exists $j > i$ where $b$ holds, and all points in between $i$ and $j$ satisfy $a$. Also,
$t_j - t_i \in I$.  In an oversampling $\rho'=(Y_1,0)(Y_2,t'_1)\dots (Y_z,t'_z)$, 
by the definition in section \ref{remove-past}, let $g(\rho'(i+s))=\rho(i)$,  
and let $g(\rho'(j+m))=\rho(j)$, $s, m \geq 0$. 
Let $l_1, l_2, \dots, l_{j-i}$ be points such that $i+s < l_1 < l_2 < \dots < l_{j-i}=j+m$ 
and $g(\rho'(l_1))=\rho(i+1), g(\rho'(l_2))=\rho(i+2), \dots, g(\rho'(l_{j-i}))=\rho(j)$. 
Then, by the definition of oversampling 
in section \ref{remove-past}, we have
\begin{enumerate}
 \item $Y_{l_d} \cap \Sigma=X_{i+d}$, for $1 \leq d \leq j-i$, and 
 \item Points $\rho'(h_g)$, $l_g+1 < h_g < l_{g+1}-1$ are 
 non action points, for $1 \leq g \leq j-i-1$. 
\item $t'_{i+l}=t_i, t'_{j+m}=t_j$. Thus, $t'_{j+m}-t'_{i+l}=t_j - t_i$.
\end{enumerate}
Hence, $\rho,i \models a \until_I b$ iff $\rho', i+l \models (act \Rightarrow a)\until_I(act \wedge b)$.  
A similar result holds for past formulae. The argument for the general case follows from the base case above.
The converse can be argued in a similar way.

\section{Converse of Lemma \ref{remove-pastinf}}
\label{pastinf-c}
Conversely, assume $\rho' \models \nu$. Let $\rho' \models \neg  \hat{X}_{[l, \infty)}$. 
 Then, there is a point $i$ such that $\rho', i \models act$ and 
 $\rho', i \nvDash (b \Leftrightarrow \past_{[l, \infty)} (a \wedge act))$. 
 Assume that $\rho', i \models b$, but $\rho', i \nvDash \past_{[l, \infty)} (a \wedge act)$. 
 Then, all points $act$ in $[0,t_i-l]$ are marked $\neg a$. Then, 
 by $\varphi_1$, 
 \begin{enumerate}
  \item all points $act$ in the $[0,l)$ future 
 of the first $a \wedge act$ must be marked $\neg b$
 \item $\neg b \wedge \neg a$ holds at all points $act$ 
 till the first $a \wedge act$. 
 \end{enumerate}
Given the above two points, we cannot have a $b$ at $t_i$. 
Thus, $\rho',i  \nvDash \past_{[l, \infty)} (a \wedge act)$, and $\rho',i \models act$ 
gives $\rho', i \models \neg b$. 

Assume now that  $\rho', i \models \neg b \wedge act$. We
then show that $\rho', i \models \neg \past_{[l, \infty)} (a \wedge act)$. 
Assume the contrary : $\rho', i \models \past_{[l, \infty)} (a \wedge act)$. 
Then there is a point marked $a \wedge act$ in $[0,t_i-l]$. 
Then, by $\varphi_2$, all points $act$ in $[t_i, \infty)$ are marked $b$, contradicting our assumption.

\section{Eliminating $\since$}
\label{since-rem}
Given a formula $\varphi \in \mtlu$ over $\Sigma$, we first flatten the formula to obtain an equisatisfiable 
formula $\varphi_{flat}$ over $\Sigma' \supset \Sigma$. In this section, we elaborate \cite{formats11}, \cite{deepak08} 
on removing the temporal definitions 
of the form $[r \Leftrightarrow (c \since f)]$ from $\varphi_{flat}$, using future operators.
We use the shortform $\nex \varphi$ to denote $false \until \varphi$.

$[r \Leftrightarrow (c \since f)]$ will be replaced by a conjunction $\nu_r$ of the following future formulae:
\begin{itemize}
 \item $\varphi_1: \widetilde{\Box}(f \Rightarrow \nex r)$
 \item $\varphi_2 : \neg r$
 \item $\varphi_3 : \Box[(r \wedge c) \Rightarrow \nex r]$
 \item $\varphi_4 : \Box[r \wedge (\neg c \wedge \neg f) \Rightarrow \nex \neg r]$
 \item $\varphi_5 : \widetilde{\Box}[(\neg r \wedge   \neg f) \Rightarrow \nex \neg r]$
 \end{itemize}

For example, consider the formula $\varphi=(a \wedge (b \wedge (c \until_{(1,2)}[(d \since e) \wedge f])))$. 
The flattened version $\varphi_{flat}=
[(d \since e) \Leftrightarrow w_1] \wedge [w_2 \Leftrightarrow c \until_{(1,2)}[w_1 \wedge f]] \wedge 
(a \wedge b \wedge w_2)$. Replace $[(d \since e) \Leftrightarrow w_1]$ with $\nu_{w_1}$ to obtain 
the equisatisfiable formula 

$\Box(e \Rightarrow \nex w_1) \wedge 
\neg w_1 \wedge \Box[(w_1 \wedge d) \Rightarrow \nex w_1] \wedge 
\Box[w_1 \wedge (\neg d \wedge \neg e) \Rightarrow \nex \neg w_1]
\wedge  \Box[(\neg w_1 \wedge   \neg e) \Rightarrow \nex \neg w_1] \wedge 
[w_2 \Leftrightarrow c \until_{(1,2)}[w_1 \wedge f]] \wedge 
(a \wedge b \wedge w_2)$.

\section{An Example}
\label{ex}
We give an example for elimination of past operator using the technique described in Lemma \ref{past-b1}. Consider the formula $\varphi = (b\Rightarrow \past_{(1,2)}a)\until(EP)$ over $\Sigma = \{a,b\}$. Formula says that before the word ends, wherever there is a $b$ there must be an $a$ in its past with time difference in $(1,2)$. We now eliminate $\past_{(1,2)}$ as follows:
\begin{itemize}
\item \textbf{Flattening:} We construct a $\hat{\varphi}_{flat}$ over $\Sigma = \{a,b,w\}$ which is equi-satisfiable to our original formula (as given by Lemma \ref{lem:flatgen}). 
$\hat{\varphi}_{flat} = (act \Rightarrow(b \Rightarrow (w \wedge act))) \until (EP \wedge act) \wedge X$. $X = \widetilde{\Box}(act \Rightarrow (w  \Leftrightarrow (\past_{[1,2)}(a \wedge act)))$. Note that 
The formula asserts that at any old action point wherever $\past_{[1,2)}a$ is true, it is marked as $w$. Thus $w \wedge act$ acts as a witness for $\past_{[1,2)}a$ and further wherever there is a $b$ that point should be marked as $w$ which means that all the points where $b$ holds $\past_{[1,2)}a$ also holds.
\item \textbf{Canonicalization:} We construct a canonical language $L'$ over $\Sigma_2 = \Sigma \cup \{c,end_w,beg_w\}$ such that for any word $\rho \in L(\hat{\varphi}_{flat})$ there is a unique $\rho' \in L'$ such that $L'\Downarrow^{\Sigma}_{\Sigma_2} = L(\hat{\varphi}_{flat})$. 

\begin{itemize}
\item All integral time points of $\rho'$ is marked as $c$ till the end of the word and no other points are marked as $c$.
\begin{quote}
$C_1 = c \wedge \widetilde{\Box}(c \Rightarrow)c \wedge \Box[ (c \wedge \neg EP) \Rightarrow
\Box_{(0,1)}\neg c \wedge[\fut_{[1,1]} c \vee \fut_{(0,1)}EP]]$
\end{quote}
\item  From every first occurrence of $a$ in an interval of the form $[x,x+1)$ where $x\in I_+ \cup\{0\}$ after exactly 1 time unit, $end_w$ holds.
\begin{quote}
 $C_2 = \widetilde{\Box}[(c \wedge \widetilde{\fut}_{[0,1)}(a \wedge act))
 \Rightarrow 
 [(act \Rightarrow \neg a) \widetilde{\until}_{[0,1)} 
 ((a\wedge act) \wedge [\widetilde{\fut}_{[1,1]}end_w \vee \widetilde{\fut}_{[0,1)}EP])]] \wedge \Box_{[0,1)}\neg end_w$
 \end{quote}
\item From every last occurrence of $a$ in an interval of the form $[x,x+1)$ where $x\in I_+ \cup\{0\}$ after exactly 2
time units, $beg_w$ holds.
\begin{quote}
$C_3 = \widetilde{\Box}[[c \wedge \widetilde{\fut}_{[0,1)}(a\wedge act)] \Rightarrow$
 $~ \widetilde{\fut}_{[0,1)} \{(a\wedge act) \wedge [((act \Rightarrow \neg a
 ) \wedge \neg c)\until c] \wedge $\\
 $(\fut_{[2,2]}beg_w \vee \widetilde{\fut}_{[0,2)}EP)\} \wedge \Box_{[0,2)}\neg beg_w$
\end{quote}
\item $beg_w$ and $end_w$ holds no where else.
$C_4 = c \wedge \widetilde{\Box}_{[0,1)}(act \Rightarrow \neg a) 
 \Rightarrow \widetilde{\Box}_{[1, 2)}\neg end_b \wedge  \widetilde{\Box}_{[2, 3)}\neg beg_b\wedge c \wedge \widetilde{\fut}_{[0,1)} x \Rightarrow 
  (\neg x \widetilde{\until}_{[0,1)}[x \wedge (\neg x \wedge \neg c)
  \until_{(0,1)} (c \vee EP)])$ for $x \in \{end_b,beg_b\}$
  \end{itemize}
  All $\rho'$ which cannot be described as above is not in the $L'$.
  \item \textbf{Eliminating Past Operator}: In this step we eliminate past operator by constructing a formula $\psi$ which when replaces $X$, results in the same timed language. Note that $X$ exactly identifies which old action points should be marked as $w$ and which should not be marked as $w$.
  \begin{itemize}
  \item \textbf{Marking points as $w$:} According to $X$ all those old action points having $a$ in their past $(1,2)$ should be marked as $w$. Which means that from any point which is marked $a$ all the old action points should be marked $w$.
  \begin{quote}
  $Y_1 = \widetilde{\Box}[(a\wedge act) \Rightarrow \Box_{[1,2)}(act \Rightarrow b)]$
  \end{quote}
  \item \textbf{Avoiding all other points to be marked as $b$} Here we restrict the marking of $b$ to only those points which are marked by $Y_1$. By main lemma \ref{past-b1} we give following formulas, which restrict this behavior. In brief, all those points from $beg_b$ to $end_b$ of the same interval should not be marked as $b$. 
  For $x= \neg(end_b \wedge c \wedge EP)$, 
  \begin{quote}
  $Y_2 = \widetilde{\Box}[(c\wedge \Box_{[0,1)} \neg beg_b) \Rightarrow 
    ((act \Rightarrow \neg b) \wedge x) \widetilde{\until} (end_b \vee c \vee EP)]\\
    \wedge 
    \widetilde{\Box}\{(c \wedge [\neg end_b
    \until_{[0,1)} beg_b])
    \Rightarrow$\\
    $~~~~~~~~~~~~~~~~~~~~~~~~~~~~~\fut_{[0,1)}[(beg_b \wedge 
    (((act \Rightarrow \neg b) \wedge x)\until (end_b \vee c \vee EP)))] \}$
  \end{quote}
  \end{itemize}
\end{itemize}
The temporal projection $X$ can be replaced by $\psi = C_1\wedge C_2\wedge C_3\wedge C_4\wedge Y_1\wedge Y_2$. Note that $\psi$ is pure $\mtl$ formula such that 
$\delta = (act \Rightarrow(b \Rightarrow (w \wedge act))) \until (EP \wedge act) \wedge \psi$ equisatisfiable to $\varphi$.


\section{Proof of Lemma \ref{game:proof}}
\label{game:sec}
We prove that the $\mtluns, \mtlsns$ are strictly less expressive than $\mtlfull$ using EF Games. 
We omit the game strategies here and give the candidate formula and pair of words.

\noindent (i) $\mtlfutpw \nsubseteq \mtluns$\\
We consider a formula in $MTL^{pw}[\fut_I]$,  $\varphi= \fut_{(0,1)}\{a \wedge \neg\fut_{[1,1]}(a\vee b)\}$. For an $n$-round game, 
consider the words $w_1=W_aW_b$ and $w_2=W_aW'_b$ with 
\begin{itemize}
 \item $W_a=(a, \delta)(a, 2 \delta) \dots (a,i\delta-\kappa)\underline{(a,i \delta)} \dots (a, n \delta)$
\item $W_b=(b, 1+\delta)(b, 1+2\delta) \dots (b,1+i\delta-\kappa)\underline{(b, 1+i\delta)} \dots (b, 1+n\delta)$
\item $W'_b=(b, 1+\delta)(b, 1+2\delta) \dots (b, 1+(i-1)\delta)(b,1+i\delta-\kappa)(b, 1+i\delta)(b, 1+(i+1)\delta) \dots (b, 1+n\delta)$ 
 \end{itemize}

$w_1 \models \varphi$, but $w_2 \nvDash \varphi$. 
The underlined $b$ in $W_b$ shows that there is a $b$ at distance 1 from $a$; however,
this is not the case with $W'_b$. 
The key observation for duplicator's win in an $\until_{NS}, \since_I$ game is that (a) any non-singular future move of spoiler can be mimicked by the duplicator from $W_aW_b$ or $W_a W'_b$ (b)
for any singular past move made by spoiler on $W_aW_b$, duplicator 
has  a reply from $W_aW'_b$. The same holds for any singular past move of spoiler made from 
$W_aW'_b$.\\

\noindent (ii) $\mtlfutp \nsubseteq \mtlsns$\\
 We consider a formula in $MTL^{pw}[\fut_I]$, $\phi' =  \fut \{b\wedge \neg \past_{[1,1]}(a \vee b)\}$. 
   We show that there is no way to express this formula in  $\mtlsns$. This is symmetrical to (i). For an $n$ round game, 
   consider the words 
   $w_1=W_aW_b$ and $w_2=W'_aW_b$ with 
\begin{itemize}
 \item $W_a=(a, \delta)(a, 2 \delta) \dots  (a, (i-1)\delta)(a, i\delta - \kappa)(a,i \delta)\dots (a, n \delta)$
 \item $W'_a=(a, \delta)(a, 2 \delta) \dots  (a, (i-1)\delta)(a,i \delta)\dots (a, n \delta)$
\item $W_b=(b, 1+\delta)(b, 1+2\delta) \dots (b, 1+(i-1)\delta)\underline{(b,1+i\delta-\kappa)}
(b, 1+i\delta) \dots (b, 1+n\delta)$
 \end{itemize}
   
 $w_1 \nvDash \varphi', w_2 \models \varphi'$.   
 The underlined $b$ in $W_b$ shows that there is an $a$ at past distance 1 
 in $W_a$, but not in $W'_a$. 
 The key observation for duplicator's win 
 in an $n$-round $\until_I, \since_{NS}$ game is that (a) any non-singular past move by spoiler from $W_a, W_b$ 
 or from $W'_a, W_b$ can be answered by duplicator, (b) for any singular future move 
   made by spoiler on $W_a, W_b$, duplicator 
has  a reply from $W'_a, W_b$. The same holds for any singular future move of spoiler made from 
$W'_a, W_b$.\\

\noindent (iii) $\mitlfp \nsubseteq \mtlu$.
We consider the $\mitlfp$ formula $\varphi''=\fut_{(1,2)}[a \wedge \neg \past_{(1,2)}a]$, and show that there is no  
way to express it using $\until_I, \since$. For an $n$ round game, consider the words
$w_1=W_1W_2$ and $w_2=W_1W'_2$ with
\begin{itemize}
 \item $W_1=(a, 0.5+\epsilon) \dots (a, 0.5+n\epsilon)(a,0.9+\epsilon)\dots(a,0.9+n\epsilon)$
 \item $W_2=\underline{(a, 1.5)}(a,1.6+\epsilon)(a, 1.6+2 \epsilon) \dots (a, 1.6+n \epsilon)$
 \item $W'_2=(a,1.6+\epsilon)(a, 1.6+2 \epsilon) \dots (a, 1.6+n \epsilon)$
\end{itemize}
for a very small $\epsilon >0$. Clearly, $w_1 \models \varphi'', w_2 \nvDash \varphi''$. 
The underlined $a$ in $W_2$ shows the $a$ in (1,2) which has no $a$ 
in $\past_{(1,2)}$. The key observation for duplicator's win 
 in an $n$-round $\until_I, \since$ game is that (a) when spoiler picks any position 
 in $W_1$, duplicator can play copy cat, (b) when spoiler  picks $(a, 1.5)$ in $W_2$ 
 as part of a future $(0,1)$ move from $W_1$, duplicator picks $0.9+n \epsilon$ in $W'_2$. All until, since 
 moves from the configuration $[(a.1.5),(a,0.9+n\epsilon)]$ are symmetric.

\section{Lemma \ref{undec2:proof}: All Formulae}
\label{undec2:sec}
 The main challenge in encoding Minsky Machine is to avoid insertion errors in the consecutive configuration. To avoid such insertion errors, 
we make use of the $\past_I$ operator. 
We make use of the shorthand $B$ for $\bigvee_{i\in {1,\ldots,n}}b_{i}$.

\begin{enumerate}
 \item Insertion of $b_{i_j}$ exactly at points $(2k+1)j$, $j \in \N$:

\begin{quote}
$
 \varphi_{0}\ =\ b_1\ \wedge \widetilde{\Box}\{B \wedge \fut B \Rightarrow \fut_{[2k+1,2k+1]}B\} \wedge \widetilde{\Box}\{B \Rightarrow (\neg \fut_{(0,2k+1)} B) \}
$
\end{quote} 
\item Intervals with no $a$: Intervals $((2k+1)j+2w,(2k+1)j+2w+1)$
have no $a$'s for $0 \leq w \leq k$.

\begin {quote}
$
 \varphi_{1}\ =\widetilde{\Box}\{B\ \Rightarrow\ \bigwedge_{w \in \{0,\ldots,k\}} (\widetilde{\Box}_{[2w,2w+1]}(\neg a)). 
$
\end {quote}
\end {enumerate}
We define macros for copying, incrementing and decrementing counters.
\begin{itemize}
 \item $COPY_i$: Every $a$ occurring in the current interval has a copy at a future distance 2k+1, 
and every $a$ occurring in the next interval has an $a$ at a past distance 2k+1. This ensures the absence 
of insertion errors.
  \begin{quote}
    $COPY_i = \Box_{(2i-1,2i)}[(a\Rightarrow \fut_{[2k+1,2k+1]} a)] \wedge \Box_{((2k+1)+2i-1,(2k+1)+2i)}[(a\Rightarrow \past_{[2k+1,2k+1]} a)]$
  \end{quote}
 
\item $INC_i$: All the $a$'s in the current configuration are copied to the next configuration, 
at a future distance 2k+1; every $a$ except the last one in the next configuration 
has an $a$ at past distance 2k+1.
  \begin{quote}  
 $INC_i = 
 \begin{array}[t]{l}  
   (\Box_{(2i-1,2i)}\{ ~[a \Rightarrow \fut_{[2k+1,2k+1]} a] \\
   \wedge [(a \wedge \neg \fut_{(0,1)}a) \Rightarrow (\fut_{(2k+1,2k+2)}a 
   \wedge \Box_{(2k+1,2k+2)}(a \Rightarrow \Box_{(0,1)}(false)))] ~\}\\
     ~~~~~~~~ 
   \wedge \Box_{((2k+1)+2i-1,(2k+1)+2i)} [(a \wedge \fut_{[0,1]}(a)) \Rightarrow \past_{[2k+1,2k+1]}a])
 \end{array}
 $
  \end{quote}

\item $DEC_i$: All the $a$'s in the current configuration,  except the last one,
have a copy at future distance 2k+1. All the $a$'s in the next configuration have a copy at past distance 2k+1.
\begin{quote}
$DEC_i = \Box_{(2i-1,2i)}\{ [(a \wedge \fut_{(0,1)}a)\Rightarrow \fut_{[2k+1,2k+1]} a] 
\wedge [(a \wedge \neg F_{[0,1]}a) \Rightarrow \neg F_{[2k+1,2k+2]} a]\} \wedge 
\Box_{((2k+1)+2i-1,(2k+1)+2i)}[(a\Rightarrow \past_{[2k+1,2k+1]} a)]$.
\end{quote}
\end{itemize}
Using the macros, we define formulae for increment, decrement and conditional jumps.
\begin{enumerate}
\item $p_x$: If $C_i=0$ goto $p_y$, else goto $p_z$
\begin{quote}
$ 
 \varphi^{x,i=0}_{3}\ =\ \widetilde{\Box}\{b_x \Rightarrow (\bigwedge_{i\in \{1,\ldots,n\}} COPY_i \wedge 
 [\Box_{(2i-1,2i)}(\neg a) \Rightarrow( \fut_{[2k+1,2k+1]}b_y)]\wedge [\fut_{(2i-1,2i)}(a) \Rightarrow(\fut_{[2k+1,2k+1]}b_{z})]\}
 $
\end{quote}
\item $p_x$: $Inc(C_i)$ goto $p_y$
\begin{quote}
$ 
 \varphi^{x,inc_i}_{3}\ =\  \widetilde{\Box}[b_x \Rightarrow (\bigwedge_{j\in \{0,\ldots,k\},j\ne i}COPY_j \wedge \fut_{[2k+1,2k+1]}b_{y} \wedge INC_i)]
$
\end {quote}
\item $p_x$: $Dec(C_i)$ goto $p_y$
\begin{quote}
$ 
 \varphi^{x,dec_i}_{3}\ =\  \widetilde{\Box}[b_x \Rightarrow (\bigwedge_{j\in \{0,\ldots,k\},j\ne i}COPY_j \wedge \fut_{[2k+1,2k+1]}b_{y} \wedge DEC_i)]
$
\end {quote}

\item No instructions are executed after HALT:
\begin{quote}
$ 
 \varphi_{2}\ =\  \widetilde{\Box}[b_n \Rightarrow \Box_{[2k+1,\infty)}(false)]
$
\end {quote}
\item Initial Configuration:  
\begin{quote}
$
 \varphi_4\ =\ b_1 \wedge \Box_{(0,2k+1)}(\neg (B \wedge a))
$
\end {quote}
\item Mutual Exclusion:- At any point of time, exactly one event takes place.
\begin{quote}
$
  \varphi_5\ =\ \bigwedge_{y \in \Sigma_{\cal M}}(y \Rightarrow \neg(\bigvee_{x \in \Sigma_{\cal M} \setminus \{y\}}(x))
$
\end{quote}

\item Termination: The HALT instruction will be seen sometime in the future.
\begin {quote}
$
 \varphi_6= \widetilde{\fut} b_n
$
\end {quote}
 
\end{enumerate}
The final formula we construct is $\varphi_{\cal M} = \bigwedge_{i=0}^6 \varphi_i$,  where $\varphi_3$ is the conjunction 
of formulae $\varphi^{inc_i}_{3}, \varphi^{dec_i}_{3}, \varphi^{i=0}_{3}$, $i \in \{1,2, \dots,k\}$.

\section{Syntax and Semantics of $\mathsf{MTL}^c$}
\label{cont}
\noindent{\bf Continuous Semantics} : 
Continuous time  MTL formulae are typically evaluated over \emph{timed state sequences} (TSS), where 
a system is assumed to continue in a state until a state change takes place. Here we interpret over timed words, 
but now a formulae can be asserted at any arbitrary time point. 
A small change here is that the atomic formula $a \in \Sigma$ can only hold at an action point 
labeled $A$ with $a \in A$ and not for an interval of time  after the event happens. 
Given a timed word $\sigma$, and an MTL formula $\varphi$, in the continuous semantics, 
the temporal connectives of $\varphi$ quantify over the whole time domain $\R_{\geq 0}$.

For an alphabet $\Sigma$, a timed word $\rho=(\sigma, \tau)=(A_1,t_1)\dots(A_n,t_n)$, 
a time $t \in \R_{\geq 0}$, 
 and an MTL formula $\varphi$, the satisfaction of $\varphi$ at time $t$ 
of $\rho$ is denoted $(\rho, t) \models \varphi$, and is defined as follows:
Let $\rho(t_i)=A_i$.

\noindent
\begin{tabular}{l c l}
$\rho, t \vDash a$ 				& \; $\leftrightarrow$ \;& $a \in \rho(t)$ \\
$\rho, t \vDash \neg \psi$ 			&$\leftrightarrow$& $\rho,t \nvDash \psi$ \\
$\rho, t \vDash \psi_1 \wedge \psi_2 $ 		&$\leftrightarrow$& $\rho,t \vDash \psi_1$ and
$\rho, t \vDash \psi_2 $ \\
$\rho, t\vDash \psi_1 \until_I \psi_2 $ 	&$\leftrightarrow$& $\exists t' \in t +  I \wedge t'>t, \rho,t'
\vDash \psi_2$ and $\forall t'' \in (t,t')$, $\rho,t'' \vDash \psi_1$ \\
$\rho, t\vDash \psi_1 \since_I \psi_2 $ 	&$\leftrightarrow$& $\exists t' \in t - I \wedge t'<t, 
\rho,t' \vDash \psi_2$ and $\forall t'' \in (t',t)$, $\rho,t'' \vDash \psi_1$ \\
\end{tabular}

We say that $\rho \models \varphi$ iff $\rho, 0 \models \varphi$. Let $L(\varphi)=\{\rho \mid \rho, 0 \models \varphi\}$.

\subsection{Undecidability of Continuous time Logic $\mtlfut$}
\label{undec:mtlc}
In this section, we show undecidability of $\mtlfut$, by encoding 2 counter machines.  
\begin{enumerate}
 \item Copy counter exactly: both action and non-action points are copied.
 \[
     COPY^c_C =  \Box_{[1,2]}[(\neg action)\Rightarrow \fut_{[5,5]} (\neg action) \wedge (a\Rightarrow \fut_{[5,5]} a)]
\]
  \item Increment counter by one exactly: 
\[
   INC^c_C = \Box_{[1,2]}
    \begin{array}[t]{l}   
       \{[(\neg action \wedge \fut_{(0,1]}a)\Rightarrow (\fut_{[5,5]} (\neg action))] \\
           ~~\wedge [(a \wedge \fut_{(0,1]}a)\Rightarrow (\fut_{[5,5]} a))] \\ 
          ~~ \wedge [(a \wedge \neg \fut_{(0,1]}a)\Rightarrow (\fut_{[5,5]} a \wedge \fut_{(5,6)}a 
\wedge \Box_{(5,6]}(a \Rightarrow \Box_{(0,1)}(\neg action)))]\}
  \end{array}
\]

\item Decrement counter by one exactly:
\[
   DEC^c_C =  \Box_{[1,2]}
\begin{array}[t]{l}   
   \{[(\neg action \wedge \fut_{(0,1]}a)\Rightarrow \fut_{[5,5]} (\neg action)]\\
    ~~ \wedge [(a \wedge \fut_{(0,1]}a)\Rightarrow \fut_{[5,5]} a] \\
    ~~ \wedge [(a \wedge \neg \fut_{(0,1]}a)\Rightarrow (\Box_{[5,6]} \neg action)] \}
\end{array}
\]
\end{enumerate}
The other formulae needed can be obtained from section \ref{minsky} 
with $k=2$.

Let the resulting version of formula be called $\varphi''_{\cal M}$. 
\begin{lemma}
 \label{undec1:proof}
Let ${\cal M}$ be a two counter Minsky machine and let $\varphi''_{\cal M} \in \mtlfut$ be
the formula as above. Then, ${\cal M}$ halts iff  $\varphi''_{\cal M}$ is satisfiable.
\end{lemma}
The proof of equivalence of $L(\cal M)$ and $L(\varphi''_{\cal M})$ is similar to Lemma \ref{undec2:proof}.

\end{document}